\documentclass[final]{siamltex}

\usepackage{todonotes}
\usepackage{comment}
\usepackage{url}
\usepackage{multicol}

\let\diag=\relax

\newtheorem{aside}[theorem]{Aside}

\usepackage{dgleich-setup}
\usepackage[notheorems]{dgleich-math}
\graphicspath{{.}{./figures/}{../figures/}}

\let\oldthebibliography=\thebibliography
\let\oldendthebibliography=\endthebibliography

\usepackage[square]{natbib}
\renewcommand{\cite}{\citep}

\let\thebibliography=\oldthebibliography
\let\endthebibliography=\oldendthebibliography

  \providecolor{theblue}{RGB}{0,0,180}%
  \RequirePackage[colorlinks,pdfdisplaydoctitle
      ,citecolor=theblue
      ,linkcolor=theblue
      ,urlcolor=theblue
      ,hyperfootnotes=false,pagebackref]{hyperref}%
\renewcommand*{\backref}[1]{}%
\renewcommand*{\backrefalt}[4]{%
  \ifcase #1 %
    No citations.%
  \or
    Cited on page #2.%
  \else
    Cited on pages #2.%
  \fi
}%

\newcommand{\mPbar}{\mat{\bar{\mP}}}
\newcommand{\mPsub}{\mat{\bar{\mP}}}

\newcommand{\vfhat}{\vec{\hat{f}}}

\title{PageRank Beyond the Web}
\author{David F.~Gleich}

\begin{document}

\maketitle

\begin{abstract}
 Google's PageRank method was developed to evaluate
 the importance of web-pages via their link structure.
 The mathematics of PageRank, however, are entirely
 general and apply to any graph or network in any 
 domain. Thus, PageRank is now regularly used 
 in bibliometrics, social and information network
 analysis, and for link prediction and recommendation.
 It's even used for systems analysis of road networks,
 as well as biology, 
 chemistry, neuroscience, and physics. We'll see the 
 mathematics and ideas that unite these diverse applications.
\end{abstract}

\begin{keywords} 
PageRank, Markov chain
\end{keywords}

\pagestyle{myheadings}
\thispagestyle{plain}
\markboth{D.~F.~Gleich}{PAGERANK BEYOND THE WEB}

\section{Google's PageRank}
Google created PageRank to address a problem they encountered with their search engine for the world wide web~\cite{brin1998-anatomy,page1999-pagerank}. Given a search query from a user, they could immediately find an immense set of web pages that contained virtually the exact same words as the user entered. Yet, they wanted to incorporate a measure of a page's importance into these results to distinguish highly recognizable and relevant pages from those that were less well known. To do this, Google designed a system of scores called PageRank that used the link structure of the web to determine which pages are important. While there are many derivations of the PageRank equations~\cite{langville2006-book,pan2004-cross-modal-discovery,higham2005-pinball}, we will derive it based on a hypothetical random web surfer. Upon visiting a page on the web, our random surfer tosses a coin. If it comes up heads, the surfer randomly clicks a link on the current page and transitions to the new page. If it comes up tails, the surfer \emph{teleports} to a -- possibly random -- page independent of the current page's identity. Pages where the random surfer is more likely to appear based on the web's structure are more important in a PageRank sense. 

More generally, we can consider random surfer models on a graph with an arbitrary set of nodes, instead of pages, and transition probabilities, instead of randomly clicked links. The teleporting step is designed to model an external influence on the importance of each node and can be far more nuanced than a simple random choice. Teleporting \emph{is} the essential distinguishing feature of the PageRank random walk that had not appeared in the literature before~\cite{Vigna-2009-spectral}. It ensures that the resulting importance scores always exist and are unique. It also makes the PageRank importance scores easy to compute. 

These features: simplicity, generality, guaranteed existence, uniqueness, and fast computation are the reasons that PageRank is used in applications far beyond its origins in Google's web-search. (Although, the success that Google achieved no doubt contributed to additional interest in PageRank!)  In biology, for instance, new microarray experiments churn out thousands of genes relevant to a particular experimental condition. Models such as GeneRank~\cite{morrison2005-generank} deploy the exact same motivation as Google, and almost identical mathematics in order to assist biologists in finding and ordering genes related to a microarray experiment or related to a disease. Throughout our review, we will see applications of PageRank to biology, chemistry, ecology, neuroscience, physics, sports, and computer systems. 

Two uses underlie the majority of PageRank applications. In the first, PageRank is used as a network centrality measure~\cite{koschutzki2005-centrality}. A network centrality score yields the importance of each node in light of the entire graph structure. And the goal is to use PageRank to help understand the graph better by focusing on what PageRank reveals as important. It is often compared or contrasted with a host of other centrality or graph theoretic measures. 
These applications tend to use global, near-uniform teleportation behaviors.

In the second type of use, PageRank is used to illuminate a region of a large graph around a target set of interest; for this reason, we call the second use a localized measure. It is also called personalized PageRank based on PageRank's origins in the web. Consider a random surfer in a large graph that periodically teleports back to a single start node. If the teleportation is sufficiently frequent, the surfer will never move far from the start node, but the frequency with which the surfer visits nodes before teleporting reveals interesting properties of this localized region of the network.  Because of this power, teleportation behaviors are much more varied for these localized applications.  

\section{The mathematics of PageRank} \label{sec:math}
There are many slight variations on the PageRank problem, yet there is a core definition that applies to the almost all of them. It arises from a generalization of the random surfer idea.  Pages where the random surfer is likely to appear have large values in the stationary distribution of a Markov chain that, with probability $\alpha$, randomly transitions according to the link structure of the web, and with probability $1-\alpha$ teleports according to a \emph{teleportation distribution vector} $\vv$, where $\vv$ is usually a uniform distribution over all pages. In the generalization, we replace the notion of ``transitioning according to the link structure of the web'' with ``transitioning according to a stochastic matrix $\mP$.'' This simple change divorces the mathematics of PageRank from the web and forms the basis for the applications we discuss. Thus, it abstracts the random surfer model from the introduction in a relatively seamless way. Furthermore, the vector $\vv$ is a critical modeling tool that distinguishes between the two typical uses of PageRank. For centrality uses, $\vv$ will resemble a uniform distribution over all possibilities; for localized uses, $\vv$ will focus the attention of the random surfer on a region of the graph.

Before stating the definition formally, let us fix some notation. Matrices and vectors are written in bold, Roman letters ($\mA, \vx$), scalars are Greek or indexed, unbold Roman $(\alpha, A_{i,j})$. The vector $\ve$ is the column vector of all ones, and all vectors are column vectors. 
 
Let $P_{i,j}$ be the probability of transitioning from page $j$ to page $i$. (Or more generally, from ``thing $j$'' to ``thing $i$''.) The stationary distribution of the PageRank Markov chain is called the PageRank vector $\vx$. It is the solution of the eigenvalue problem: 
\begin{equation} \label{eq:pr-eigen}
(\alpha \mP + (1-\alpha) \vv \ve^T) \vx = \vx. 
\end{equation}
Many take this eigensystem as the definition of PageRank~\cite{langville2006-book}. We prefer the following definition instead: 
\begin{definition}[The PageRank Problem] \label{def:pagerank} Let $\mP$ be a column-stochastic matrix where all entries are non-negative and the sum of entries in each column is 1. Let $\vv$ be a column stochastic vector $(\ve^T \vv = 1)$, and let $0 < \alpha < 1$ be the teleportation parameter. Then the PageRank problem is to find the solution of the linear system 
\begin{equation} \label{eq:pr}
(\mI - \alpha \mP) \vx = (1-\alpha) \vv,
\end{equation}
where the solution $\vx$ is called the PageRank vector. 
\end{definition}

The eigenvector and linear system formulations are equivalent if we seek an eigenvector $\vx$ of \eqref{eq:pr-eigen} with $\vx \ge 0$ and $\ve^T \vx = 1$, in which case: 
\[ \vx = \alpha \mP \vx + (1-\alpha) \vv \ve^T \vx = \alpha \mP \vx + (1-\alpha) \vv \quad \Leftrightarrow \quad (\mI - \alpha \mP) \vx = (1-\alpha) \vv. \]
We prefer the linear system because of the following reasons. In the linear system setup, the existence and uniqueness of the solution is immediate: the matrix $\eye - \alpha \mP$ is a diagonally dominant M-matrix. The solution $\vx$ is non-negative for the same reason.  Also, there is only one possible normalization of the solution: $\vx \ge 0$ and $\ve^T \vx = 1$. Anecdotally, we note that, among the strategies to solve PageRank problems, those based on the linear system setup are both more straightforward and more effective than those based on the eigensystem approach. And in closing, \citet{page1999-pagerank} describe an iteration more akin to a linear system than an eigenvector.

Computing the PageRank vector $\vx$ is simple. The humble iteration
\[ \vx\itn{k+1} = \alpha \mP \vx\itn{k} + (1-\alpha) \vv \qquad \text{ where } \qquad \vx\itn{0} = \vv \text{ or } \vx\itn{0} = 0 \]
is equivalent both to the power method on \eqref{eq:pr-eigen} and the Richardson method on \eqref{eq:pr}, and more importantly, it has excellent convergence properties when $\alpha$ is not too close to 1. To see this fact, note that the true solution $\vx = \alpha \mP \vx + (1-\alpha) \vv$ and consider the error after a single iteration: 
\[ \vx - \vx\itn{k+1} = \underbrace{[\alpha \mP \vx + (1-\alpha) \vv]}_{\text{the true solution $\vx$}}  - \underbrace{[\alpha \mP \vx\itn{k} + (1-\alpha) \vv]}_{\text{the updated iterate $\vx\itn{k+1}$}} = \alpha \mP (\vx - \vx\itn{k}). \]
Thus, the following theorem characterizes the error after $k$ iterations from two different starting conditions:
\begin{theorem} \label{thm:pagerank-error}
Let $\alpha, \mP, \vv$ be the data for a PageRank problem to compute a PageRank vector $\vx$. Then the error after $k$ iterations of the update $\vx\itn{k+1} = \alpha \mP \vx\itn{k} + (1-\alpha) \vv$ is: 
\begin{compactenum} 
\item if $\vx\itn{0} = \vv$,  then $\normof[1]{\vx - \vx\itn{k}} \le \normof[1]{\vx - \vv} \alpha^k \le 2 \alpha^k$; or 
\item if $\vx\itn{0} = 0$, then the error vector $\vx - \vx\itn{k} \ge 0$ for all $k$ and $\normof[1]{\vx - \vx\itn{k}} = \ve^T (\vx - \vx\itn{k}) = \alpha^k$.
\end{compactenum}
\end{theorem}
Common values of $\alpha$ range between $0.1$ and $0.99$; hence, in the worst case, this method needs at most $3656$ iterations to converge to a global $1$-norm error of $2^{-52} \approx 10^{-16}$ (because $\alpha^{3656} \le 2^{-53}$ to account for the possible factor of $2$ if starting from $\vx\itn{0} = \vv$). For the majority of applications we will see, the matrix $\mP$ is sparse with fewer than $10,000,000$ non-zeros; and thus, these solutions can be computed efficiently on a modern laptop computer. 

\begin{aside}
Although this theorem seems to suggest that $\vx\itn{0} = 0$ is a superior choice, practical experience suggests that starting with $\vx\itn{0} = \vv$ results in a faster method. This may be confirmed by using a computable bound on the error based on the residual. Let $\vr\itn{k} = (1-\alpha) \vv - (\mI - \alpha \mP) \vx\itn{k} = \vx\itn{k+1} - \vx\itn{k}$ be the residual after $k$ iterations. We can use  $\normof[1]{\vx - \vx\itn{k}} = \normof[1]{(\mI - \alpha \mP)^{-1} \vr\itn{k}} \le \frac{1}{1-\alpha} \normof[1]{\vr\itn{k}}$ in order to check for early convergence. 
\end{aside}

This setup for PageRank, where the choice of $\mP$, $\vv$, and $\alpha$ vary by application, applies broadly as the subsequent sections show. However, in many descriptions, authors are not always careful to describe their contributions in terms of a column stochastic matrix $\mP$ and distribution vector $\vv$. Rather, they use the following pseudo-PageRank system instead:
\begin{definition}[The pseudo-PageRank problem] Let $\mPbar$ be a column sub-stochastic matrix where $\bar{P}_{i,j} \ge 0$ and $\ve^T \mPsub \le \ve^T$ element-wise. Let $\vf$ be a non-negative vector, and let $0 < \alpha < 1$ be a teleportation parameter. Then the pseudo-PageRank problem is to find the solution of the linear system
\begin{equation} \label{eq:pseudo-pr}
(\mI - \alpha \mPbar) \vy = \vf
\end{equation}
where the solution $\vy$ is called the pseudo-PageRank vector.
\end{definition}

Again, the pseudo-PageRank vector always exists and is unique because $\mI- \alpha \mPbar$ is also a diagonally dominant M-matrix. \Citet{boldi2007-traps} was the first to formalize this definition and distinction between PageRank and pseudo-PageRank, although they used the term PseudoRank and the normalization $(\mI - \alpha \mPbar) \vy = (1-\alpha) \vf$; some advantages of this alternative form are discussed in Section~\ref{sec:limit}. The two problems are equivalent in the following formal sense (which has an intuitive understanding explained in Section~\ref{sec:standard-rw}, Strongly Preferential PageRank): 
\begin{theorem} \label{thm:pagerank-equiv}
Let $\vy$ be the solution of a pseudo-PageRank system with $\alpha $, $\mPbar$ and $\vf$. Let $\vv = \vf/(\ve^T \vf)$. Then if $\vy$ is renormalized to sum to $1$, that is $\vx = \vy / (\ve^T \vy)$, then $\vx$ is the solution of a PageRank system with $\alpha$, $\mP = \mPbar + \vv \vc^T$, and $\vv$, where $\vc^T = \ve^T - \ve^T \mPsub \ge 0$ is a  correction vector to make $\mPsub$ stochastic. 
\end{theorem}
\begin{proof}
First note that $\alpha, \mP$, and  $\vv$ is a valid PageRank problem. This is because $\vf$ is non-negative and thus $\vv$ is column stochastic by definition, and also $\mP$ is column stochastic because  $\vc \ge 0$ (hence $\mP \ge 0$) and $\ve^T \mP = \ve^T \mPsub + \vc^T = \ve^T$.  Next, note that the solution of the PageRank problem for $\vx$ satisfies: 
\[ \vx = \alpha \mPsub \vx + \alpha \vv \vc^T \vx + (1-\alpha) \vv = \alpha \mPsub \vx + \gamma \vf \qquad \text{ where } \qquad \gamma = \frac{\alpha \vc^T \vx + (1-\alpha) }{ \ve^T \vf }. \]
Hence $(\mI - \alpha \mPsub) \vx = \gamma \vf$ and so $\vx = \gamma \vy$. But, we know that $\ve^T \vx = 1$ because $\vx$ is a solution of a PageRank problem, and the theorem follows. 
\end{proof}

The importance of this theorem is it shows that underlying any pseudo-PageRank system is a true PageRank system in the sense of Definition~\ref{def:pagerank}. The difference is entirely in terms of the normalization of the solution -- which was demonstrated by \citet{DelCorso-2004-sparse,berkhin2005-survey,delcorso2005-sparse-system}.  The result of Theorem~\ref{thm:pagerank-error} also applies to solving the pseudo-PageRank system, albeit with the following revisions:
\begin{theorem}
Let $\alpha, \mPsub, \vf$ be the data for a pseudo-PageRank problem to compute a pseudo-PageRank vector $\vy$. Then the error after $k$ iterations of the update $\vy\itn{k+1} = \alpha \mPsub \vy\itn{k} + \vf$ is: 
\begin{compactenum} 
\item if $\vy\itn{0} = \frac{1}{1-\alpha} \vf$, then  $\normof[1]{\vy - \vy\itn{k}} \le \normof[1]{\vy - \vf} \alpha^k \le \frac{2\ve^T \vf }{1-\alpha}    \alpha^k$; or 
\item if $\vy\itn{0} = 0$, then the error vector $\vy - \vy\itn{k} \ge 0$ for all $k$ and $\normof[1]{\vy - \vy\itn{k}} = \ve^T (\vy - \vy\itn{k}) \le \alpha^k$.
\end{compactenum}
\end{theorem}
\begin{aside}
The error progression proceeds at the same rate for both PageRank and pseudo-PageRank. This can be improved for pseudo-PageRank if the vector $\vc^T = \ve^T - \ve^T \mPsub > 0$ (element-wise). In such cases, then we can derive an equivalent system with a smaller value of $\alpha$ and a suitably rescaled matrix $\mPsub$. 
\end{aside}

These formal results represent the mathematical foundations of all of the PageRank systems that arise in the literature (with a few technical exceptions that we will study in Section~\ref{sec:generalizations}). The results depend only on the construction of a stochastic matrix or sub-stochastic matrix, a teleportation distribution, and a parameter $\alpha$. Thus, they apply generally and have no intrinsic relationship back to the original motivation of PageRank for the web.  Each type of PageRank problem has a unique solution that always exists, and the two convergence theorems justify that simple algorithms for PageRank converge to the unique solutions quickly. These are two of the most attractive features of PageRank.

One final set of mathematical results is important to understand the behavior of localized PageRank; however, the precise statement of these results requires a lengthy and complicated diversion into graph partitioning, graph cuts, and spectral graph theory. Instead, we'll state this a bit informally. Suppose that we solve a localized PageRank problem in a large graph, but the nodes we select for teleportation lie in a region that is somehow isolated, yet connected to the rest of the graph. Then the final PageRank vector is large only in this isolated region and has small values on the remainder of the graph. This behavior is exactly what most uses of localized PageRank want: they want to find out what is nearby the selected nodes and far from the rest of the graph.  Proving this result involves spectral graph theory, Cheeger inequalities, and localized random walks -- see \citet{andersen2006-local} for more detail. Instead, we illustrate this theory with Figure~\ref{fig:localized}. 

\begin{figure}
\hfil\includegraphics{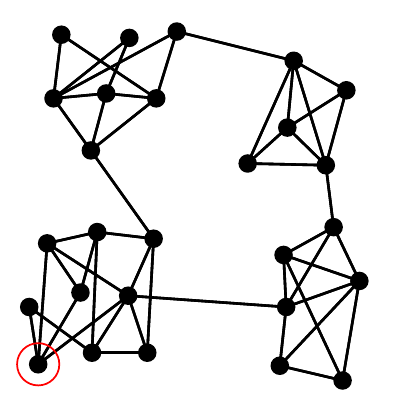}%
\hfil\includegraphics{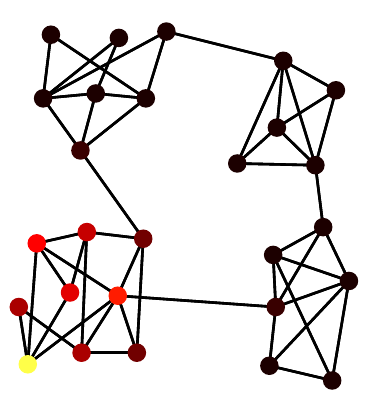}%
\hfill\par%

\caption{An illustration of the empirical properties of localized PageRank vectors with teleportation to a single node in an isolated region. In the graph at left, the teleportation vector is the single circled node. The PageRank vector is shown as the node color in the right figure. PageRank values remain high within this region and are nearly zero in the rest of the graph. Theory from \citet{andersen2006-local} explains when this property occurs.}
\label{fig:localized}
\end{figure}

Next, we will see some of the common constructions of the matrices $\mP$ and $\mPsub$ that arise when computing PageRank on a graph. These justify that PageRank is also a simple construction.

\section{PageRank constructions}
\label{sec:constructions}

When a PageRank method is used within an application, there are two common motivations. In the centrality case, the input is a graph representing relationships or flows between a set of things -- they may be documents, people, genes, proteins, roads, or pieces of software -- and the goal is to determine the expected importance of each piece in light of the full set of relationships and the teleporting behavior. This motivation was Google's original goal in crafting PageRank. In the localized case, the input is also the same type of graph, but the goal is to determine the importance relative to a small  subset of the objects. In either case, we need to build a stochastic or sub-stochastic matrix from a graph. In this section, we review some of the common constructions that produce a PageRank or pseudo-PageRank system. For a visual overview of some of the possibilities, see Figures~\ref{fig:pagerank-theory}~and~\ref{fig:pagerank-constructions}.

\begin{figure}[ht]
\includegraphics[width=\linewidth]{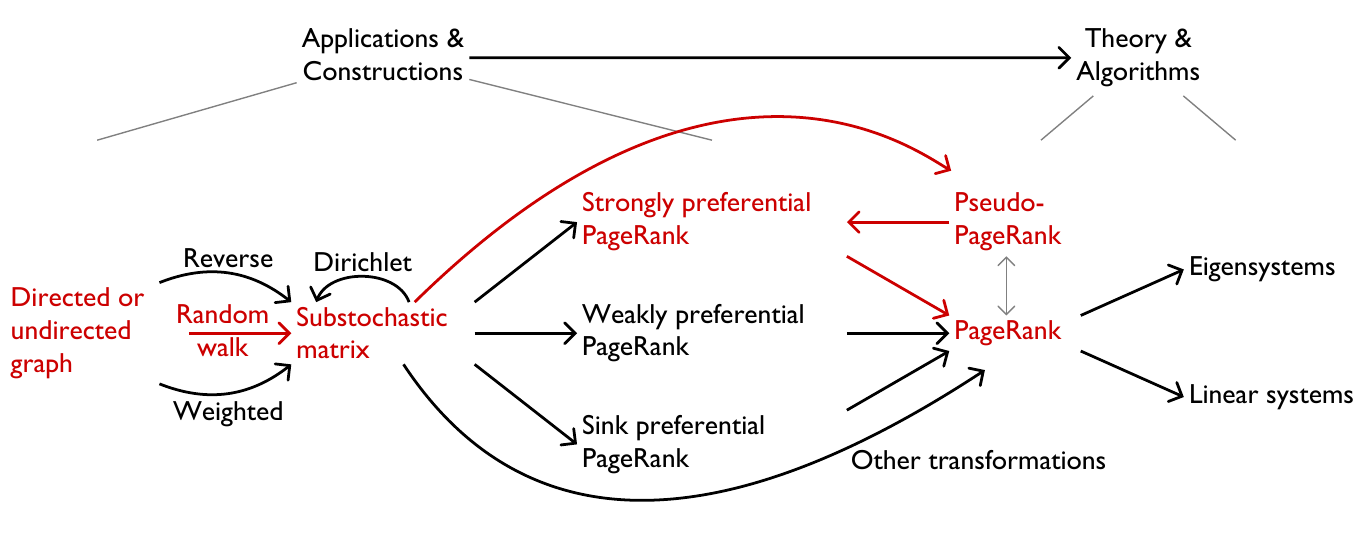}
\caption{An overview of PageRank constructions and how they relate.  The vast majority of PageRank applications fall somewhere on the red path.}
\label{fig:pagerank-theory}
\end{figure}

\paragraph{Notation for graphs and matrices}
Let $\mA$ be the adjacency matrix for a graph where we assume that the vertex set is $V = \{ 1, \ldots, n \}$.  The graph could be directed, in which case $\mA$ is non-symmetric, or undirected, in which case $\mA$ is symmetric. The graph could also be weighted, in which case $A_{i,j}$ gives the positive weight of edge $(i,j)$. Edges with zero weight are assumed to be irrelevant and equivalent to edges that are not present. For such a graph, let $\vd$ be the vector of node out-degrees, or equivalently, the vector of row-sums: $\vd = \mA \ve$. The matrix $\mD$ is simply the diagonal matrix with $\vd$ on the diagonal. Weighted graphs are extremely common in applications when the weights reflect a measure of the \emph{strength} of the relationships between two nodes. 

\subsection{The standard random walk} \label{sec:standard-rw}
In the standard construction of PageRank, the matrix $\mP$ represents a uniform random walk operation on the graph $\mA$. When the graph is weighted, the simple generalization is to model a non-uniform walk that chooses subsequent nodes with probability proportional to the connecting edge's weight. The elements of $\mPsub$ are rather similar between the two cases: 
\[ \bar{P}_{j,i} = \frac{A_{i,j}}{\sum_{k} A_{i,k}}  = \frac{A_{i,j}}{d_i} = \begin{array}{l} \text{probability of taking the transition} \\ \text{from $i$ to $j$ via a random walk step}. \end{array} \]
Notice two features of this construction. First, we transpose between $j,i$ and $i,j$. This is because $A_{i,j}$ indicates an edge from node $i$ to node $j$, whereas the probability transition matrix element $i,j$ indicates that node $i$ \emph{can be reached} via node $j$. Second, we have written $\mPbar$ and $\bar{P}_{j,i}$ here because there may be nodes of the graph with \emph{no outlinks}. These nodes are called \emph{dangling nodes}.  Dangling nodes complicate the construction of stochastic matrices $\mP$ in a few ways because we must specify a behavior for the random walk at these nodes in order to fully specify the stochastic matrix. 

As a matrix formula, the standard random walk construction is:
\[ \mPsub = \mA^T \mD^{+}. \]
Here, we have used the \emph{pseudo-inverse} of the degree matrix to ``invert'' the diagonal matrix  in light of the dangling nodes with $0$ out-degrees.
Let $\vc^T$ be the sub-stochastic correction vector. For the standard random walk construction, $\vc^T$ is just an indicator vector for the dangling nodes: 
\[ c_i = 1 - \sum_{k} \bar{P}_{k,i} = \begin{cases} 1 & \text{node $i$ is dangling} \\ 0 & \text{otherwise.} \end{cases} \]

We shall now see a few ideas that turn these sub-stochastic matrices into fully stochastic PageRank problems.

\paragraph{Strongly Preferential PageRank}
Given a directed graph with dangling nodes, the standard random walk construction produces the sub-stochastic matrix $\mPsub$ described above. If we had just used this matrix to solve a pseudo-PageRank problem with a stochastic teleportation vector $\vf = (1-\alpha) \vv$, then, by Theorem~\ref{thm:pagerank-equiv}, the result is equivalent up to normalization to computing PageRank on the matrix: 
\[ \mP = \mPsub + \vc \vv^T. \]
This construction models a random walk that transitions according to the distribution $\vv$ when visiting a dangling node. This behavior reinforces the effect of the teleportation vector $\vv$, or preference vector as it is sometimes called. Because of this reinforcement,   \citet{boldi2007-traps} called the construction $\mP = \mPsub + \vc \vv^T$ a \emph{strongly preferential PageRank} problem. Again, many authors are not careful to \emph{explicitly} choose a correction to turn the sub-stochastic matrix into a stochastic matrix. Their lack of choice, then, \emph{implicitly} chooses the strongly preferential PageRank system. 


\paragraph{Weakly Preferential PageRank \famp\ Sink Preferential PageRank} 
\Citet{boldi2007-traps} also proposed the \emph{weakly preferential PageRank} system. In this case, the behavior of the random walk at dangling nodes is adjusted \emph{independently} of the choice of teleportation vector. For instance, \citet{langville04-deeper} advocates transitioning uniformly from dangling nodes. In such a case, let $\vu = \ve / n$ be the uniform distribution vector, then a weakly preferential PageRank system is: 
\[ \mP = \mPsub + \vc \vu^T. \]
We note that another choice of behavior is for the random walk to remain at dangling nodes until it moves away via a teleportation step: \[ \mP = \mPsub + \diag(\vc). \]
We call this final method \emph{sink preferential PageRank}. These systems are less common.  These choices should be used when the matrix $\mP$ models some type of information or material flow that must be decoupled from the teleporting behavior.

\subsection{Reverse PageRank} \label{sec:reverse-pagerank}
In reverse PageRank, we compute PageRank on the transposed graph $\mA^T$. This corresponds to reversing the direction of each edge $(i,j)$ to be an edge $(j,i)$. Reverse PageRank is often used to determine \emph{why} a particular node is important rather than \emph{which} nodes are important~\cite{Fogaras-2003-where,gyongyi2004-trustrank,Bar-Yossef2008-reverse-PageRank}. Intuitively speaking, in reverse PageRank, we model a random surfer that follows in-links instead of out-links. Thus, large reverse PageRank values suggest nodes that can reach many nodes in the graph. When these are localized, they then provide evidence for why a node has large PageRank. 

\subsection{Dirichlet PageRank}
Consider a PageRank problem where we wish to fix the importance score of a subset of nodes~\cite{Chung-2011-Dirichlet}. Let $S$ be a subset of nodes such that $i \in S$ implies than $v_i = 0$. A Dirichlet PageRank problem seeks a solution of PageRank where each node $i$ in $S$ is fixed to a \emph{boundary} value $b_i$. Formally, the goal is to find $\vx$:
\[ (\mI - \alpha \mP) \vx = (1-\alpha) \vv \qquad \text{ where } \qquad x_i = b_i \text{ for } i \in S. \]
These problems reduce to solving a pseudo-PageRank system. Consider a block partitioning of $\mP$ based on the set $S$ and the complement set of vertices $\bar{S}$: 
\[ \mP = \bmat{ \mP_{S,S} & \mP_{S,\bar{S}} \\ \mP_{\bar{S},S} & \mP_{\bar{S},\bar{S}} }. \]
Then the Dirichlet PageRank problem is 
\[ \bmat{\mI & 0 \\ -\alpha \mP_{\bar{S},S} & \mI - \alpha \mP_{\bar{S},\bar{S}} } \bmat{ \vb  \\ \vx_{\bar{S}} } = (1-\alpha) \bmat{ 0 \\ \vv_{\bar{S}} }. \]
This system is equivalent to a pseudo-PageRank problem with $\mPsub = \mP_{\bar{S},\bar{S}}$ and $\vf = (1-\alpha) \vv_{\bar{S}} + \alpha \mP_{\bar{S},S} \vb$.

\subsection{Weighted PageRank} \label{sec:weighted}
In the standard random walk construction for PageRank on an unweighted graph, the probability of transitioning from node $i$ to any of it's neighbors $j$ is the same: $1/d_i$. Weighted PageRank~\cite{Xing-2004-weighted-pagerank,Jiang-2009-ranking-spaces} alters this assumption such that the walk preferentially visits high-degree nodes. Thus, the probability of transitioning from node $i$ to node $j$ depends on the degree of $j$ relative to the total sum of degrees of all $i$'s neighbors. In our notation, if the input is adjacency matrix $\mA$ with degree matrix $\mD$, then the sub-stochastic matrix $\mPsub$ is given by the non-uniform random walk construction on the weighted graph with adjacency matrix $\mW = \mA \mD$, that is, $\mPsub = \mD \mA^T \diag(\mA \mD \ve)^{+}$. More generally, let $\mD_W$ be a non-negative weighting matrix. It could be derived from the graph itself based on the out-degree, in-degree, or total-degree (the sum of in- and out-degree), or from some external source.  
Then $\mPsub = \mD_W \mA^T \diag(\mA \mD_W \ve)^{-1}$.
\emph{Let us note that weighted PageRank uses a specific choice of weights for the prior importance of each node; the setting here already adapts seamlessly to edge-weighted graphs.}

\subsection{PageRank on an undirected graph} \label{sec:undir-pr}
One final construction is to use PageRank on an undirected graph. Those familiar with
Markov chain theory often find this idea puzzling at first. A uniform random walk on a
connected, undirected graph has a well-known, unique stationary distribution~\cite[][is a good numerical treatment of such issues]{Stewart-1994-markov}: 
\[ \underbrace{ \mA^T \mD^{-1} }_{\mP} \vx = \vx \text{ is solved by } \vx = \mD \ve / (\ve^T \vd). \]
This works because both the row and column sums of $\mA$ and $\mA^T$ are identical, and the resulting construction is a \emph{reversible Markov chain}~\cite[][is a good reference on this topic]{Aldous-2014-reversible}. If $\alpha < 1$, then the PageRank Markov chain \emph{is not} a reversible Markov chain even on an undirected graph, and hence, has no simple stationary distribution. PageRank vectors of undirected graphs, when combined with carefully constructed teleportation vectors $\vv$, yield important information about the presence of small isolated regions in the graph~\cite{andersen2006-local,Gleich-2014-alg-anti-diff}; formally these results involve graph cuts and small conductance sets. These vectors are most useful when the teleportation vector is far away from the uniform distribution, such as the case in Figure~\ref{fig:localized} where the graph is undirected.

\begin{aside} Of course, if the teleportation distribution $\vv = \mD \ve / (\ve^T \vd)$, then the resulting chain is reversible. The PageRank vector is then equal to $\vv$ itself. There are also specialized PageRank-style constructions that preserve reversibility with more interesting stationary distributions~\cite{Avrachenkov-2010-uniform}.
\end{aside}

\begin{figure}
\footnotesize
\begin{minipage}[m]{0.3\linewidth}
\centering
\includegraphics{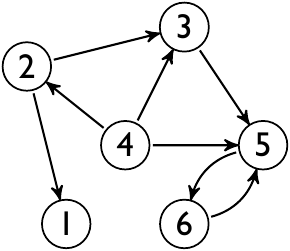}
\end{minipage}%
\begin{minipage}[m]{0.7\linewidth}
\centering
$ \mA = 
\bmat{  
0 & 0 & 0 & 0 & 0 & 0 \\
1 & 0 & 1 & 0 & 0 & 0 \\ 
0 & 0 & 0 & 0 & 1 & 0 \\
0 & 1 & 1 & 0 & 1 & 0 \\
0 & 0 & 0 & 0 & 0 & 1 \\
0 & 0 & 0 & 0 & 1 & 0} 
\quad
\vd = \bmat{
0 \\ 2 \\ 1 \\ 3 \\ 1 \\ 1}
\;
\vc = \bmat{
1 \\ 0 \\ 0 \\ 0 \\ 0 \\ 0}$
\end{minipage} \medskip 

\begin{minipage}[m]{0.3\linewidth}
\centering
A directed graph
\end{minipage}%
\begin{minipage}[m]{0.7\linewidth}
\centering
The adjacency matrix, degree vector, and correction vector
\end{minipage}
\bigskip 

\begin{tabularx}{\linewidth}{XXX}
\emph{\normalsize Random walk} & \emph{\normalsize Strongly preferential} & \emph{\normalsize Weakly preferential} \\ \addlinespace
$\mPsub = \sbmat{
\vphantom{1/6} \; 0 \; & 1/2 & \; 0 \;  & \; 0 \;    &  \; 0 \;  &  \; 0 \;  \\ 
\vphantom{1/6} 0 & 0   & 0 & 1/3 & 0 & 0 \\ 
\vphantom{1/6} 0 & 1/2 & 0 & 1/3 & 0 & 0 \\ 
\vphantom{1/6} 0 & 0   & 0 & 0   & 0 & 0 \\ 
\vphantom{1/6} 0 & 0   & 1 & 1/3 & 0 & 1 \\ 
\vphantom{1/6} 0 & 0   & 0 & 0   & 1 & 0 }$ 
& $\mP = \sbmat{
\vphantom{1/6} \; 0 \;  & 1/2 & \; 0 \;  & \; 0 \;    & \; 0 \;  & \; 0 \;  \\ 
\vphantom{1/6} 0 & 0   & 0 & 1/3 & 0 & 0 \\ 
\vphantom{1/6} 1/3 & 1/2 & 0 & 1/3 & 0 & 0 \\ 
\vphantom{1/6} 1/3 & 0   & 0 & 0   & 0 & 0 \\ 
\vphantom{1/6} 1/3 & 0   & 1 & 1/3 & 0 & 1 \\ 
\vphantom{1/6} 0 & 0   & 0 & 0   & 1 & 0 }$
& $\mP = \sbmat{
 1/6  & 1/2 & \; 0 \;  & \; 0 \;    & \; 0 \;  & \; 0 \;  \\ 
 1/6 & 0   & 0 & 1/3 & 0 & 0 \\ 
 1/6 & 1/2 & 0 & 1/3 & 0 & 0 \\ 
 1/6 & 0   & 0 & 0   & 0 & 0 \\ 
 1/6 & 0   & 1 & 1/3 & 0 & 1 \\ 
 1/6 & 0   & 0 & 0   & 1 & 0 }$ 
 \\ \addlinespace 
 $\mPsub = \mA^T \mD^{+}$ & $\mP = \mPsub + \vv \vc^T$  & $\mP = \mPsub + \vu \vc^T$ \\
& & $\qquad \vu \not= \vv$ \\ \addlinespace \addlinespace 

\emph{\normalsize Reverse} & \emph{\normalsize Dirichlet }& \emph{\normalsize Weighted} \\ \addlinespace
$\mPsub = \sbmat{
\vphantom{1/6} \; 0 \; & \; 0 \; & \; 0 \; & \; 0 \; & \; 0 \; & \; 0 \; \\
\vphantom{1/6} 1 & 0 & 1/2 & 0 & 0 & 0 \\ 
\vphantom{1/6} 0 & 0 & 0 & 0 & 1/3 & 0 \\
\vphantom{1/6} 0 & 1 & 1/2 & 0 & 1/3 & 0 \\
\vphantom{1/6} 0 & 0 & 0 & 0 & 0 & 1 \\
\vphantom{1/6} 0 & 0 & 0 & 0 & 1/3 & 0}  $ 
&
$\mPsub = \sbmat{
\vphantom{1/6}  \; 0 \;  & \; 0 \; & 1/3 & \; 0 \; & \; 0\; \\ 
\vphantom{1/6} 1/2 & 0 & 1/3 & 0 & 0 \\ 
\vphantom{1/6} 0   & 0 & 0   & 0 & 0 \\ 
\vphantom{1/6} 0   & 1 & 1/3 & 0 & 1 \\ 
\vphantom{1/6} 0   & 0 & 0   & 1 & 0 }$ \newline $~\qquad S = \{ 2, 3, 4, 5, 6 \}$
& 
$\mPsub = \sbmat{
\vphantom{1/6} \; 0 \;  & 1/4 & \; 0 \;  & \; 0 \;    & \; 0 \;  & \; 0 \;  \\ 
\vphantom{1/6} 0 & 0   & 0 & 3/10 & 0 & 0 \\ 
\vphantom{1/6} 0 & 3/4 & 0 & 3/10 & 0 & 0 \\ 
\vphantom{1/6} 0 & 0   & 0 & 0   & 0 & 0 \\ 
\vphantom{1/6} 0 & 0   & 1 & 4/10 & 0 & 1 \\ 
\vphantom{1/6} 0 & 0   & 0 & 0   & 1 & 0 }$
\\ \addlinespace
$\mPsub = \mA \diag( \mA^T \ve )^{+}$ & $\mPsub = \mPsub_{\bar{S},\bar{S}}$ & $\mPsub = (\mD_W \mA^T) \diag(\mA \mD_W \ve)^{+}$ \\
  & $\qquad S \subset V$ & $\mD_W$ is a diagonal weighting matrix, e.g.~total degree here\\

\end{tabularx}
\caption{A directed graph and some of the different PageRank constructions on that graph. For the stochastic constructions, we have $\vv^T = [ \, 0 \; 0 \; \tfrac{1}{3} \; \tfrac{1}{3} \; \tfrac{1}{3} \; 0 \, ]$ and $\vu = \ve/n$. Note that node $4$ is dangling in the reverse PageRank construction. For the weighted construction, the total degrees are $[ \, 1 \; 3 \; 3 \; 3 \; 4 \; 2 \, ]$.
}
\label{fig:pagerank-constructions}
\end{figure}


\section{PageRank applications}

When PageRank is used within applications, it tends to acquire a new name. We will see:
\begin{multicols}{4}
\parindent=18pt
\parskip=3pt
\footnotesize

GeneRank

ProteinRank

IsoRank

MonitorRank

BookRank 

TimedPageRank

CiteRank 

AuthorRank

PopRank

FactRank

ObjectRank

FolkRank

ItemRank

BuddyRank

TwitterRank

HostRank

DirRank

TrustRank

BadRank

VisualRank

\end{multicols}

The remainder of this section explores the uses of PageRank within different domains. It is devoted to the most interesting and diverse uses and should not, necessarily, be read linearly.  Our intention is not to cover the full details, but to survey the diversity of applications of PageRank. We recommend returning to the primary sources for additional detail. 
\begin{multicols}{2}
\parindent=0pt
\parskip=3pt
\footnotesize
\centering 

Chemistry $\cdot$ \S\ref{sec:chemistry}

Biology $\cdot$ \S\ref{sec:biology}

Neuroscience $\cdot$ \S\ref{sec:neuroscience}

Engineered systems $\cdot$ \S\ref{sec:engineered}

Mathematical systems $\cdot$ \S\ref{sec:mathematics}

Sports $\cdot$ \S\ref{sec:sports}

Literature $\cdot$ \S\ref{sec:literature}

Bibliometrics $\cdot$ \S\ref{sec:bibliometrics}

Databases \famp\ Knowledge systems $\cdot$ \S\ref{sec:knowledge}

Recommender systems $\cdot$ \S\ref{sec:recommender}

Social networks $\cdot$ \S\ref{sec:social}

The web, redux $\cdot$ \S\ref{sec:web}
\end{multicols}

\subsection{PageRank in chemistry} \label{sec:chemistry}
The term ``graph'' arose from ``chemico-graph'' or a picture of a chemical structure~\cite{Sylvester-1878-Graph}. Much of this chemical terminology remains with us today. For instance, the valence of a molecule is the number of potential bonds it can make. The valence of a vertex is synonymous with its degree, or the number  of connections it makes in the graph. It is fitting, then, that recent work by \citet{Mooney-2012-pagerank-chemistry} uses PageRank to study molecules in chemistry. In particular, they use PageRank to assess the change in a network of molecules linked by hydrogen bonds among water molecules. Given the output of a molecular dynamics simulation that provides geometric locations for a solute in water, the graph contains edges between the water molecules if they have a potential hydrogen bond to a solute molecule. The goal is to assess the hydrogen bond potential of a solvent. The PageRank centrality scores using uniform  teleportation with $\alpha = 0.85$ are strongly correlated with the degree of the node -- which is expected -- but the deviance of the PageRank score from the degree identifies important outlier molecules with smaller degree than many in their local regions. The authors compare the networks based the PageRank values with and without a solute to find structural differences.

\subsection{PageRank in biology \famp\ bioinformatics: GeneRank, ProteinRank, IsoRank} \label{sec:biology}

Biology and bioinformatics are currently awash in network data. Some of the most interesting applications of PageRank arise when it is used to study these networks. Most of these applications use PageRank to reveal localized information about the graph based on some form of external data.

\paragraph{GeneRank}
Microarray experiments are a measurement of whether or not a gene's expression is promoted or repressed in an experimental condition. Microarrays estimate the outcomes for thousands of genes simultaneously in a few experimental conditions. The results are extremely noisy.  GeneRank~\cite{morrison2005-generank} is a PageRank-inspired idea to help to denoise them. The essence of the idea is to use a graph of known relationships between genes to find genes that are highly related to those promoted or repressed in the experiment, but were not themselves promoted or repressed. Thus, they use the microarray expression results as the teleportation distribution vector for a PageRank problem on a network of known relationships between genes. The network of relationships between genes is undirected, unweighted with a few thousand nodes. This problem uses a localized teleportation behaviour and, experimentally, the best choice of $\alpha$ ranges between $0.75$ and $0.85$. Teleporting is used to focus the search.

\paragraph{Finding correlated genes} This same idea of using a network of known relationships in concert with an experiment encapsulates many of the other uses of PageRank in biology. \Citet{Jiang-2009-generank} use a combination of PageRank and BlockRank~\cite{kamvar2003-blockrank,Kamvar-2010-personalized} on tissue-specific protein-protein interaction networks in order to find genes related to type 2 diabetes. The teleportation is provided by 34 proteins known to be related to that disease with $\alpha = 0.92$. 

\Citet{Winter-2012-CancerRank} use PageRank to study pancreatic ductal adenocarcinoma, a type of cancer responsible for 130,000 deaths each year, with a particularly poor prognosis (2\% mortality after five years). They identified seven genes that better predicted patient survival than all existing tools, and validated this in a clinical trial. One curious feature is that their teleportation parameter was small, $\alpha = 0.3$. This was chosen based on a cross-validation strategy in a statistically rigorous way. The particular type of teleportation they used was based on the correlation between the expression level of a gene and the survival time of the patient. 

\paragraph{ProteinRank}
The goal of ProteinRank~\cite{freschi2007-proteinrank} is similar, in spirit, to GeneRank. Given an undirected network of protein-protein interactions and human-curated functional annotations about what these proteins do, the goal is to find proteins that may share a functional annotation. Thus, the PageRank problem is, again, a localized use. The teleportation distribution is given by a random choice of nodes with a specific functional annotation. The PageRank vector reveals proteins that are highly related to those with this function, but do not themselves have that function labeled. 

\paragraph{Protein distance} Recall that the solution of a PageRank problem for a given teleportation vector $\vv$ involves solving $(\eye - \alpha \mP) \vx = (1-\alpha) \vv$. The resolvent matrix $\mX = (1-\alpha) (\eye - \alpha \mP)^{-1}$ corresponds to computing PageRank vectors that teleport to every individual node. The entry $X_{i,j}$ is the value of the $i$th node when the PageRank problem is localized on node $j$. One interpretation for this score is the PageRank that node $j$ contributes to node $i$, which has the flavor of a similarity score between node $i$ and $j$. \Citet{Voevodski-2009-spectral} base an affinity measure between proteins on this idea. Formally, consider an undirected, unweighted protein-protein interaction network. Compute the matrix $\mX$ for $\alpha = 0.85$, and the affinity matrix $\mS = \min(\mX, \mX^T)$. (For an undirected graph, a quick calculation shows that $\mX^T = \mD^{-1} \mX \mD$.) For each vertex $i$ in the graph, form links to the $k$ vertices with the largest values in  row of $i$ of $\mS$. These PageRank affinity scores show a much larger correlation with known protein relationships than do other affinity or similarity metrics between vertices. 

\paragraph{IsoRank}
Consider the problem of deciding if the vertices of two networks can be mapped to each other. The relationship between this problem and PageRank is surprising and unexpected; although precursor literature existed~\cite{jeh2002-simrank,blondel2004-graph-similarity}. \Citet{singh2007-matching-topology} proposes a PageRank problem to estimate how much of a match the two nodes are in a diffusion sense. They call it IsoRank based on the idea of ranking graph isomorphisms. Let $\mP$ be the Markov chain for one network and let $\mQ$ be the Markov chain for the second network. Then IsoRank solves a PageRank problem on $\mQ \kron \mP$. The solution vector $\vx$ is a vectorized form of a matrix $\mX$ where $X_{ij}$ indicates a likelihood that vertex $i$ in the network underlying $\mP$ will match to vertex $j$ in the network  underlying $\mQ$. See Figure~\ref{fig:isorank} for an example. If we have an apriori measure of similarity between the vertices of the two networks, we can add this as a teleportation distribution term. IsoRank problems are some of the largest  PageRank problems around due to the Kronecker product (e.g.~\citet{gleich2010-inner-outer} has a problem with 4 billion nodes and 100 billion edges). But there are quite a few good algorithmic approaches to tackle them by using properties of the Kronecker product~\cite{Bayati-2013-netalign} and low-rank matrices~\cite{Kollias-2011-netalign}.

The IsoRank authors consider the problem of matching protein-protein interaction networks between distinct species. The goal is to leverage insight about the proteins from a species such as a mouse in concert with a matching between mouse proteins and human proteins, based on their interactions, in order to hypothesize about possible functions for proteins in a human. For these problems, each protein is coded by a gene sequence. The authors construct a teleportation distribution by comparing the gene sequences of each protein using a tool called BLAST. They found that using $\alpha$ around $0.9$ gave the highest structural similarity between the two networks. 

\begin{figure} 
\centering
\begin{minipage}[m]{0.35\linewidth}
\vspace{0pt}\centering\footnotesize\itshape
\includegraphics{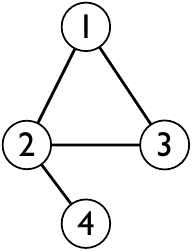}
\includegraphics{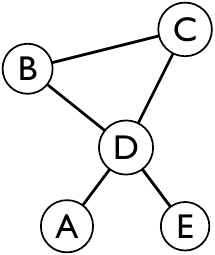}
\end{minipage}%
\begin{minipage}[m]{0.3\linewidth}
\centering
$\mP = \sbmat{
\vphantom{1/6} \; 0 \; & 1/3 & 1/2  & \; 0 \; \\ 
\vphantom{1/6} 1/2 & 0   & 1/2 & 1 \\ 
\vphantom{1/6} 1/2 & 1/3 & 0 & 0  \\ 
\vphantom{1/6} 0 &  1/3   & 0 & 0   }$

$\mQ = \sbmat{
\vphantom{1/6} \; 0 \; & \; 0 \; & \; 0 \; & 1/4 & \; 0 \; \\ 
\vphantom{1/6} 0 & 0   & 1/2 & 1/4 & 0 \\ 
\vphantom{1/6} 0 & 1/2 & 0 & 1/4 & 0  \\ 
\vphantom{1/6} 1 &  1/2   & 1/2 & 0 & 1 \\
\vphantom{1/6} 0 & 0 & 0 & 1/4 & 0  }$
\end{minipage}%
\begin{minipage}[m]{0.35\linewidth}
\centering\footnotesize\itshape
$\begin{smallmatrix}
& \arraycolsep=3pt\hspace{-3pt}\begin{matrix} \;\;\; A \;\; & \;\; B \;\; & \;\; C \;\; & \;\; D \;\; & \;\; E \;\;\; \end{matrix} \\
\begin{matrix} \vphantom{1/6}
1 \\ \vphantom{1/6}
2 \\ \vphantom{1/6}
3 \\ \vphantom{1/6}
4\end{matrix} &
\arraycolsep=3pt
\bmat{ \vphantom{1/6}
0.03 & 0.05 & 0.05 & 0.09 & 0.03 \\ \vphantom{1/6} 
0.04 & 0.07 & 0.07 & 0.15 & 0.04 \\ \vphantom{1/6} 
0.03 & 0.05 & 0.05 & 0.09 & 0.03 \\ \vphantom{1/6} 
0.02 & 0.03 & 0.03 & 0.05 & 0.02 }  
\end{smallmatrix}$
\end{minipage}

\begin{minipage}[m]{0.35\linewidth}
\centering\footnotesize\itshape
(1) Two graphs
\end{minipage}%
\begin{minipage}[m]{0.3\linewidth}
\centering\footnotesize\itshape
(b) Their stochastic matrices
\end{minipage}%
\begin{minipage}[m]{0.35\linewidth}
\centering\footnotesize\itshape
(c) The IsoRank solution
\end{minipage}
\caption{An illustration of the IsoRank problem. The solution, written here as a matrix, gives the similarity between pairs of nodes of the graph. For instance, node $2$ is most similar to node $D$. Removing this match, then nodes $1$ and $3$ are indistinguishable from $B$ and $C$. Removing these leaves node $4$ equally similar to $A$ and $E$. In this example we solved $(\eye - \alpha \mQ \kron \mP) \vx = (1-\alpha) \ve / 20$ with $\alpha = 0.85$.}
\label{fig:isorank}
\end{figure}

\subsection{PageRank in neuroscience} \label{sec:neuroscience}

The human brain connectome is one of the most important networks, about which we understand surprisingly little. Applied network theory is one of a variety of tools currently used to study it~\cite{Sporns-2011-book}. Thus, it is likely not surprising that PageRank has been used to study the properties of networks related to the connectome~\cite{Zuo-2011-centrality}.  Most recently, PageRank helped evaluate the importance of brain regions given observed correlations of brain activity. In the resulting graph, two voxels of an MRI scan are connected if the correlation between their functional MRI time-series is high. Edges with weak correlation are deleted and the remainder are retained with either binary weights or the correlation weights. The resulting graph is also undirected, and they use PageRank, combined with community detection and known brain regions, in order to understand changes in brain structure across a population of 1000 individuals that correlate with age.

Connectome networks are widely hypothesized to be hierarchically organized. Given a directed network that should express a hierarchical structure, how can we recover the order of the nodes that minimizes the discrepancy with a hierarchical hypothesis? \citet{Crofts-2011-googling} consider PageRank for this application on networks of neural connections from C.~Elegans.  They find that this gives poor results compared with other network metrics such as the Katz score~\cite{katz1953-status}, and communicability~\cite{Estrada-2008-multipartite}. In their discussion, the authors note that this result may have been a mismatch of models, and conjecture that the flow of influence in PageRank was incorrect. Literature involving Reverse PageRank (Section~\ref{sec:reverse-pagerank}) strengthens this conjecture. Let us reiterate that although PageRank models are easy to apply, they must be employed with some care in order to get the best results. 

\subsection{PageRank in complex engineered systems: MonitorRank} \label{sec:engineered}

The applications of PageRank to networks in chemistry, biology, and neuroscience are part of the process of investigating and analyzing something we do not fully understand. PageRank methods are also used to study systems that we explicitly engineered. As these engineered systems grow, they become increasingly complex, with networks and submodules interacting in unpredictable, nonlinear ways. Network analysis methods like PageRank, then, help reveal these details.  We'll see two examples: software systems and city systems. 

\paragraph{MonitorRank} 
Diagnosing root causes of issues in a modern distributed system is painstaking work. It involves repeatedly searching through error logs and tracing debugging information. MonitorRank~\cite{Kim-2013-monitorrank} is a system to provide guidance to a systems administrator or developer as they perform these activities. It returns a ranked list of systems based on the likelihood that they contributed to, or participated in, an anomalous situation. Consider the systems underlying the LinkedIn website: each service provides one or more APIs that allow other services to utilize its resources. For instance, the web-page generator uses the database and photo store. The photo store in turn uses the database, and so on. Each combination of a service and a programming interface becomes a node in the MonitorRank graph. Edges are directed and indicate the direction of function calls -- e.g. web-page to photo store. Given that an anomaly was detected in a system, MonitorRank solves a personalized PageRank problem on a weighted, augmented version of the call graph, where the weights and augmentation depend on the anomaly detected. (The construction is interesting, albeit tangential, and we refer readers to that paper for the details.) The localized PageRank scores help determine the anomaly. The graphs involved are fairly small: a few hundred to a few thousand nodes.

\paragraph{PageRank of the Linux kernel}
The Linux kernel is the foundation for an open source operating system. It has evolved over the past 20 years with contributions from nearly 2000 individuals in an effort with an estimated value of \$3 billion. As of July 2013, the Linux kernel comprised 15.8 million lines of code containing around 300,000 functions. The kernel call graph is a network that represents dependencies between functions and both PageRank and reverse PageRank, as centrality scores, produce an ordering of the most important functions in Linux~\cite{Chepelianskii-2010-SoftwareRank}. The graphs were directed with a few million edges. Teleportation was typical: $\alpha = 0.85$ with a global, uniform $\vv = \ve/n$. They find that utility functions such as \texttt{printk}, which prints messages from the kernel, and \texttt{memset}, a routine that initializes a region of memory, have the highest PageRank, whereas routines that initialize the system such as \texttt{start\textunderscore kernel} have the highest reverse PageRank.
\Citet{Chepelianskii-2010-SoftwareRank} further uses the distribution of PageRank and reverse PageRank scores to characterize the properties of a software system. (This same idea is later used for Wikipedia too,~\citealt{Zhirov-2010-wikipedia}, Section~\ref{sec:web}.)

\paragraph{Roads and Urban Spaces}
Another surprising use of PageRank is with road and urban space networks. PageRank helps to predict both traffic flow and human movement in these systems. The natural road construction employed is an interesting graph. A natural road is more or less what it means: it's a continuous path, built from road segments by joining adjacent segments together if the angle is sufficiently small and there isn't a better alternative. (For help visualizing this idea, consider traffic directions that state: ``Continue straight from High street onto Main street.'' This would mean that there is one natural road joining High street and Main street.) Using PageRank with $\alpha = 0.95$, \citet{Jiang-2008-predicting} finds that PageRank  is the best network measure in terms of predicting  traffic on the individual roads. These graphs have around 15,000 nodes and around 50,000 edges. Another group used PageRank to study Markov chain models based on the line-graph of roads~\cite{Schlote-2012-road-rank}. That is, given a graph of intersections (nodes) and roads (edges), the line graph, or dual graph, changes the role of roads to the nodes and intersections to the edges. In this context, PageRank's teleportation mirrors the behavior of starting or ending a journey on each street. This produces a different value of $\alpha$ for each node that reflects the tendency of individuals to park, or end their journey, on each street. Note that this is slightly different setup where each node has a separate teleportation parameter $\alpha$, rather than a different entry in the teleportation vector. Assuming that each street has some probability of a journey ending there, then this system is equivalent to a more general PageRank construction (Section~\ref{sec:dummy-node}). These Markov chains are used to study road planning and optimal routing in light of new constraints imposed by electric vehicles.

An urban space is the largest space of a city observable from a single vantage point. For instance, the Mission district of San Francisco is too large, but the area surrounding Dolores Park is sufficiently small to be appreciated as a whole. For the study by \citet{Jiang-2009-ranking-spaces}, an urban space is best considered as a city neighborhood or block.  The urban space network connects adjacent spaces, or blocks, if they are physically adjacent. The networks of urban spaces in London, for instance, have up to 20,000 nodes and 100,000 links. In these networks, weighted PageRank (Section~\ref{sec:weighted}) best predicts human mobility in a case study of movement within London. It outperforms PageRank, and in fact, they find that weighted PageRank with $\alpha = 1$ accounts for up to 60\% of the observed movement. Both using weighted PageRank and $\alpha = 1$ make sense for these problems -- individuals and businesses are likely to co-locate places with high connectivity, and individuals cannot teleport over the short time-frames used for the human mobility measurements. Based on the evidence here, we would hypothesize that using $\alpha < 1$ would better generalize over longer time-spans. 

\subsection{PageRank in mathematical systems} \label{sec:mathematics}

Graphs and networks arise in mathematics to abstract the properties of systems of equations and processes to relationships between simple sets.  We present one example of what PageRank reveals about a dynamical system by abstracting the phase-space to a discrete set of points and modeling transitions among them. Curiously, PageRank and its localization properties has not yet been used to study properties of Cayley graphs from large, finite groups, although closely related structures have been examined~\cite{Frahm-2012-IntRank}.

\paragraph{PageRank of symbolic images and Ulam networks}
Let $f$ be a discrete-time dynamical system on a compact state space $M$. For instance, $M$ will be the subset of $\RR^2$ formed by $[ 0, 2 \pi ] \times [ 0, 2 \pi ]$ for our example below. Consider a covering of $M$ by cells $C$. In our forthcoming example, this covering will just be a set of non-overlapping cells that form a regular, discrete partition into cells of size $2 \pi /N \times 2 \pi /N$. The symbolic image~\cite{Osipenko-2007-dynamical-systems-graphs} of $f$ with respect to $C$ is a graph where the vertices are the cells and $C_i \in C$ links to $C_j \in C$ if $x \in C_i$ and $f(x) \in C_j$. The Ulam network is a weighted approximation to this graph that is constructed by simulating $s$ starting points within cell $C_i$ and forming weighted links to their destinations $C_j$~\cite{Shepelyansky2010-Ulam}. The example studied by those authors, and the example we will consider here, is the Chirikov typical map.  
\[ \begin{aligned} 
  y_{t+1} & = \eta y_t + k \sin(x_t + \theta_t) \\
  x_{t+1} & = x_t + y_{t+1}. \end{aligned} \]
It models a kicked oscillator. We generate $T$ random phases $\theta_t$ and look at the map: 
\[ f(x,y) = (x_{T+1}, y_{T+1}) \text{ mod } 2 \pi \quad \text{ where } \quad x_1 = x, y_1 = y. \]
That is, we iterate the map for $T$ steps for each of the $T$ random phase shifts $\theta_1, \ldots, \theta_T$.
Applying the construction above with $s = 1000$ random samples from each cell yields a directed weighted graph $G$ with $N^2$ nodes and at most $N^2 s$ edges. PageRank on this graph, with uniform teleportation, yields beautiful pictures of the transient behaviors of this chaotic dynamical system; these are easy to highlight with modest teleportation parameters such as $\alpha = 0.85$ because this regime inhibits the dynamical system from converging to its stable attractors. This application is particularly useful for modeling the effects of different PageRank constructions as we illustrate in Figure~\ref{fig:chirikov}. For that figure, the graph has $262,144$ nodes and $4,106,079$ edges, $\eta = 0.99, k = 0.22, T = 10$.

\begin{figure}

\includegraphics[width=0.33\linewidth]{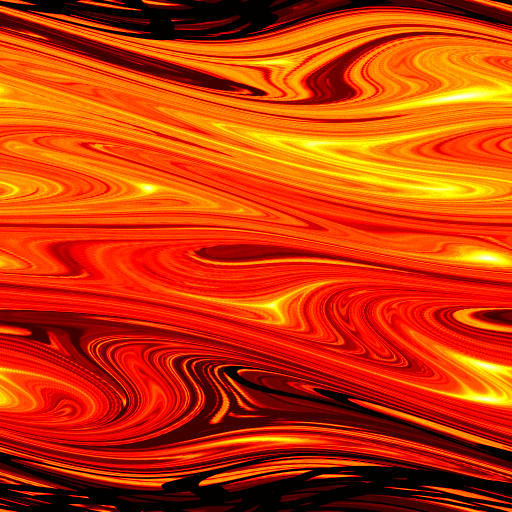}%
\includegraphics[width=0.33\linewidth]{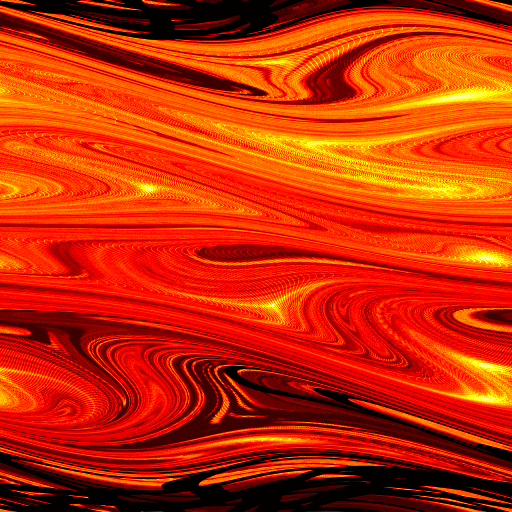}%
\includegraphics[width=0.33\linewidth]{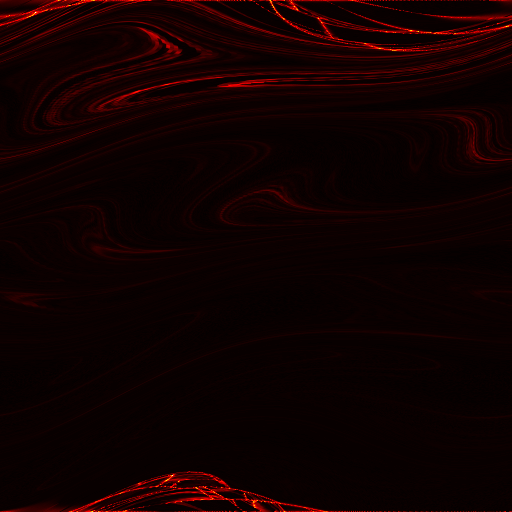}

\caption{PageRank vectors of the symbolic image, or Ulam network, of the Chirikov typical map with $\alpha = 0.9$ and uniform teleportation. From left to right, we show the standard PageRank vector, the weighted PageRank vector using the unweighted cell in-degree count as the weighting term, and the reverse PageRank vector. Each node in the graph is a point $(x,y)$, and it links to all other points $(x,y)$ reachable via the map $f$ (see the text). The graph is weighted by the likelihood of the transition. PageRank, itself, highlights both the attractors (the bright regions), and the contours of the transient manifold that leads to the attractor. The weighted vector looks almost identical, but it exhibits an interesting stippling effect. The reverse PageRank highlights regions of the phase-space that are exited quickly, and thus, these regions are dark or black in the PageRank vector.  The solution vectors were scaled by the cube-root for visualization purposes.  These figures are incredibly beautiful and show important transient regions of these dynamical systems.}
\label{fig:chirikov}
\end{figure}

\subsection{PageRank in sports} \label{sec:sports}
Stochastic matrices and eigenvector ranking methods are nothing new in the realm of sports ranking~\cite{keener1993-ranking,Callaghan-2007-random,Langville-2012-book}. One of the natural network constructions for sports is the winner network. Each team is a node in the network, and node $i$ points to node $j$ if $j$ won in the match between $i$ and $j$. These networks are often weighted by the score by which team $j$ beat team $i$. \Citet{govan2008-pagerank-football} used the centrality sense of PageRank with uniform teleportation and $\alpha = 0.85$ to rank football teams with these winner networks. The intuitive idea underlying these rankings is that of a random fan that follows a team until another team beats them, at which point they pick up the new team, and periodically restarts with an arbitrary team. In the \citet{govan2008-pagerank-football} construction, they corrected dangling nodes using a strongly preferential modification, although, we note that a sink preferential modification may have been more appropriate given the intuitive idea of a random fan. \citet{Radicchi-2011-tennisrank} used PageRank on a network of tennis players with the same construction. Again, this was a weighted network. PageRank with $\alpha = 0.85$ and uniform teleportation on the tennis network placed Jimmy Conors in the best player position. 

\subsection{PageRank in literature: BookRank} \label{sec:literature}
PageRank methods help with three problems in literature. What are the most important books? Which story paths in hypertextual literature are most likely? And what should I read next?

For the first question, \citet{Jockers-2012-authorrank} defines a complicated distance metric between books using topic modeling ideas from latent Dirichlet allocation~\cite{Blei-2003-LDA}. Using PageRank as a centrality measure on this graph, in concert with other graph analytic tools, allows Jockers to argue that Jane Austin and Walter Scott are the most original authors of the 19th century. 

Hypertextual literature contains multiple possible story paths for a single novel. Among American children of similar age to me, the most familiar would be the \emph{Choose your own adventure} series. Each of these books consists of a set of storylets; at the conclusion of a storylet, the story either ends, or presents a set of possibilities for the next story. \citet{Kontopoulou-2012-storyrank} argue that the random surfer model for PageRank maps perfectly to how users read these books. Thus, they look for the most probable storylets in a book. For this problem, the graphs are directed and acyclic, the stochastic matrix is normalized by outdegree, and we have a standard PageRank problem. They are careful to model a weakly preferential PageRank system that deterministically transitions from a terminal (or dangling) storylet back to the start of the book. Teleporting is uniform in their experiments. They find that both PageRank and a ranking system they derive give useful information about the properties of these stories.

\paragraph{Books \famp\ tags: BookRank}
Traditional library catalogs use a carefully curated set of index terms to indicate the contents of books. These enabled content-based search prior to the existence of fast full-text search engines. Social cataloging sites such as LibraryThing and Shelfari allow their users to curate their own set of index terms for books that they read, and easily share this information among the user sites. The data on these websites consists of \emph{books} and \emph{tags} that indicate the topics of books. BookRank, which is localized PageRank on the bipartite book-tag graph~\cite{Meng-2009-bookrank}, produces eerily accurate suggestions for what to read next. For instance, if we use teleportation to localize on Golub and van Loan's text ``Matrix Computations'', Boyd and Vandenberghe's ``Convex Optimization'', and Hastie, Tibshirani, and Friedman's ``Elements of Statistical Learning'', then the top suggestion is a book on Combinatorial Optimization by Papadimitriou and Steiglitz.  A similar idea underlies the general FolkRank system~\cite{Hotho2006-folksonomies} that we'll see shortly (Section~\ref{sec:knowledge}).

\subsection{PageRank in bibliometrics: TimedPageRank, CiteRank, AuthorRank} \label{sec:bibliometrics}

The field of bibliometrics is another big producer and consumer of network ranking methods, starting with seminal work by Garfield on aggregating data into a citation network between journals~\cite{Garfield-1955-citation,Garfield-1963-citation} and proceeding through \citet{Pinski-1976-influence}, who defined a close analogue of PageRank. In almost all of these usages, PageRank is used as a centrality measure to reveal the most important journals, papers, and authors.

\paragraph{Citations among journals} The citation network Garfield originally collected and analyzed is the journal-journal citation network. It is a weighted network where each node is a journal and each edge is the number of citations between articles of the journals. \textsc{ISI}'s impact factor is a more refined analysis of these citation patterns. \citet{Bollen-2006-journal-status} takes \textsc{ISI}'s methods a step further and finds that a combination of the impact factor with the PageRank value in the journal citation produces a ranked list of journals that better correlates with experts' judgements.  PageRank is used as a centrality measure here with uniform teleportation and weights that correspond to the weighted citation network. The graph had around 6000 journals. The Eigenfactor system~\cite{West-2010-eigenfactor} uses a PageRank vector on the journal co-citation network with uniform teleportation and $\alpha = 0.85$ to measure the influence of a journals. It also shows these rankings on easy-to-browse website.

\paragraph{Citations among papers: TimedPageRank, CiteRank}
Moving beyond individual journals, we can also study the citation network among individual papers using PageRank. In a paper citation network, each node is an individual article and the edges are directed based on the citation. Modern bibliographic and citation databases such as arXiv and DBLP make these networks easy to construct. They tend to have hundreds of thousands of nodes and a few million edges. TimedPageRank is an idea to weight the edges of the stochastic matrix in PageRank such that more recent citations are more important. Formally, it is the solution of 
\[ (\eye - \alpha \mA^T \mD^{-1} \mW) \vx = (1-\alpha) \ve \]
where $\mW$ is a diagonal matrix with weights between $0$ and $1$ that reflects the age of the paper (1 is recent and 0 is old). The matrix $\mA^T \mD^{-1} \mW$ is column sub-stochastic and so this is a pseudo-PageRank problem. CiteRank is a subsequent idea that uses the teleportation in PageRank to increase the rank of recent articles~\cite{Walker-2007-citerank}. Thus, $v_i$ is smaller if paper $i$ is older and $v_i$ is larger if paper $i$ is more recent. The goal of both methods is to produce temporally relevant orderings that remove the bias of older articles to acquire citations.

While the previous two papers focused on how to make article importance more accurate, \citet{chen2007-finding-gems} attempts to use PageRank in concert with the number of citations to find \emph{hidden gems}. One notable contribution is the study of $\alpha$ in citation analysis: based on a heuristic argument about how we build references for an article, they recommend $\alpha = 0.5$. Moreover, they find papers with higher PageRank scores than would be expected given their citation count. These are the hidden gems of the literature. \Citet{Ma-2008-PageRank-citation} uses the same idea in a larger study and find a similar effect.  

\paragraph{Citations among authors: AuthorRank}
Another type of bibliographic network is the co-authorship graph. For each paper, insert edges among all co-authors. Thus, each paper becomes a clique in the co-authorship network. The weights on each edge are either uniform (and set to $1$), based on the number of papers co-authored, or based on another weighting construction defined in that paper. All of these constructions produce an undirected network. PageRank on this network gives a practical ranking of the most important authors~\cite{Liu-2005-co-authorship-networks}. The teleportation is uniform with $\alpha = 0.85$, or can be focused on a subset of authors to generate an area-specific ranking. Their data have a few thousand authors.  These graphs are constructions based on an underlying bipartite matrix $\mB$ that relates authors and papers. More specifically, the weighted co-authorship network is the matrix $\mB \mB^T$. Many such constructions can be related back to the matrix $\smash{\sbmat{0 & \mB \\ \mB^T & 0}}$~\cite{Dhillon-2002-spectral-co-clustering}.  We are not aware of any analysis that makes a relationship between PageRank in the bipartite graph $\smash{\sbmat{0 & \mB \\ \mB^T & 0}}$ and the weighted matrix $\mB \mB^T$.

\paragraph{Author, paper, citation networks}
Citation analysis and co-authorship analysis can, of course, be combined, and that is exactly what \citet{Fiala-2008-PageRank-bibliographic} and \citet{Jezek-2008-citation} do.
Whereas \citet{Liu-2005-co-authorship-networks} study the co-authorship network, here, they study a particular construction that joins the bipartite author-paper network to the citation network to produce an author-citation network. This is a network where author $i$ links to author $j$ if $i$ has a paper that cites $j$ where $j$ is not a co-author on that particular paper. Using $\alpha = 0.9$ and uniform teleportation produces another helpful list of the most important authors. 
In the notation of the previous paragraph, a related construction is the network with adjacency matrix \[ \mA = \bmat{ 0 & \mB \\
                \mB^T & \mC }, \] where $\mB$ is the bipartite author-paper matrix and $\mC$ is the citation matrix among papers. PageRank on these networks takes into account both the co-authorship and directed citation information, and it rewards authors that have many, highly cited papers.  The graphs studied have a few hundred thousand authors and author-author citations.

\subsection{PageRank in databases and knowledge information systems: PopRank, FactRank, ObjectRank, FolkRank} \label{sec:knowledge} 

Knowledge information systems store codified forms of information, typically as a relational database. For instance, a knowledge system about movies consists of relationships between actors, characters, movies, directors, and so own. Contemporary information systems also often contain large bodies of user-generated content through tags, ratings, and such. Ratings are a sufficiently special case that we review them in a forthcoming section (Section~\ref{sec:recommender}), but we will study PageRank and tags here. PageRank serves important roles as both a centrality measure and localized measure in networks derived from a knowledge system. We'll also present slightly more detail on four interesting applications.

\paragraph{Centrality scores: PopRank, FactRank} PageRank's role as a centrality measure in a knowledge information system is akin to its role on the web as an importance measure. For instance, the authors of PopRank~\cite{Nie-2005-PopRank} consider searching through large databases of objects -- think of academic papers -- that have their own internal set of relationships within the knowledge system -- think of co-author relationships. But these papers are also linked to by websites. PopRank uses web-importance as a teleportation vector for a PageRank vector defined on the set of object relationships. The result is a measure of object popularity biased by its web popularity. One of the challenges in using such a system is that collecting good databases of relational information is hard. FactRank helps with this process~\cite{Jain-2010-FactRank}. It is a measure designed to evaluate the importance and accuracy of a fact network. A fact is just a sentence that connects two objects, such as ``David-Gleich \emph{wrote} the-paper PageRank-Beyond-The-Web.'' These sentences come from textual analysis of large web crawls. In a fact network, facts are connected if they involve the same set of objects. Variations on PageRank with uniform teleportation provide lists of important facts. The authors of FactRank found that weighting relationships between facts and using  PageRank scores of this weighted network gave higher performance than both a baseline and standard PageRank method in the task of finding correct facts. The fact networks are undirected and have a few million nodes.

\paragraph{Localized scores: Random-walk with restart, Semi-Supervised Learning}
Prediction tasks akin to the bioinformatics usages of PageRank are standard within knowledge information systems: networks contain noisy relationships, and the task is inferring, or predicting, missing data based on these relationships. \Citet{Zhou-2003-semi-supervised} used a localized PageRank computation to infer the identity of handwritten digits from only a few examples. These problems were called \emph{semi-supervised learning on graphs} because they model the case of finding a vector over vertices (or learning a function) based on a few values of the function (supervised). It differs from the standard supervised learning problem because the graph setup implies that only predictions on the vertices are required, instead of the general prediction problem with arbitrary future inputs. In the particular study, the graph among these images is based on a radial basis function construction. For this task $\alpha = 0.99$ in the pseudo-PageRank system $(\mI - \alpha \mP) \mat{\tilde{Y}} = \mS$, where $\mS$ is a binary matrix indicating \emph{known samples} $S_{ij} = 1$ if image $i$ is known to be digit $j$. The largest value in each row of $\mY = \mD \mat{\tilde{Y}}$ gives the predicted digit for any unknown image. While these graphs were undirected, later work ~\cite{zhou2005-directed-learning} showed how to use PageRank with global teleportation, in concert with symmetric Laplacian structure defined on a directed graph~\cite{Chung2005-directed-laplacian}, to enable the same methodology on a general directed graph. 

\Citet{pan2004-cross-modal-discovery} define a random walk with restart, which is exactly a personalized PageRank system, to infer captions for a database of images. Given a set of images labeled by captions, define a graph where each image is connected to its regions, each region is connected to other regions via a similarity function, and each image is connected to the terms in its caption. A query image is a distribution over regions, and we find terms by solving a PageRank problem with this as the teleportation vector. These graphs are weighted, undirected graphs. Curiously, the authors chose $\alpha$ based on experimentation and found that $\alpha = 0.1$ or $\alpha = 0.2$ works best. They attribute the difference to the incredibly small diameter of their network. Subsequent work in the same vein showed some of the relationships with the normalized Laplacian matrix of a graph~\cite{tong2006-random-walk-restart} and returned to a larger value of $\alpha$ around $0.9$.

\paragraph{Application 1 -- Database queries: ObjectRank}
ObjectRank is an interesting type of database query~\cite{Balmin-2004-ObjectRank}. A typical query to a database will retrieve all of the rows of a specified set of tables matching a precise criteria, such as, ``find all students with a GPA of 3.5 that were born in Minnesota.'' These tables often have internal relationships -- the database schema -- that would help determine which are the most important returned results. In the ObjectRank model, a user queries the database with a textual term. The authors describe a means to turn the database objects and schema into a sub-stochastic transition matrix and define ObjectRank as the query-dependent solution of the PageRank linear system where the teleportation vector  reflects textual matches.  They suggest a great deal of flexibility with defining the weights of this matrix. For instance, there may be no natural direction for many of these links and the authors suggest differently weighting forward edges and backward edges -- their intuition is that a paper cited by many important papers is itself important, but that citing important paper papers does not transfer any importance.  They use $\alpha = 0.85$ and the graphs have a few million edges.

\paragraph{Application 2 -- Folksonomy search: FolkRank}
A more specific situation is folksonomy search. A folksonomy is a collection of objects, users, and tags. Each entry is a triplet of these three items. A user such as myself may have tagged a picture on the flickr network with the term ``sunset'' if it contained a sunset, thus creating the triplet (picture,user,``sunset''). FolkRank scores~\cite{Hotho2006-folksonomies} are designed to measure the importance of an object, tag, or user with respect to a small set of objects, tags, or users that define a topic. (This idea is akin to topic-sensitive PageRank,~\citealt{haveliwala2002-topic-pagerank}.) These scores then help reveal important objects \emph{related} to a given search, as well as the tags that relate them. The scores are based on localized PageRank scores from an undirected, tripartite weighted network. There is a wrinkle, however. The FolkRank scores are taken as the \emph{difference} between a PageRank vector computed with $\alpha = 1$ and $\alpha = 1/2$. The graph is undirected, so the solution with $\alpha = 1$ is just the weighted degree distribution. Thus, FolkRank downweights items that are important for \emph{everyone}. 

\paragraph{Application 3 -- Semantic relatedness}
The Open Directory Project, or \textsc{odp}, is a hierarchical, categorical index of web-pages that organizes them into related groups. \Citet{Bar-Yossef2008-reverse-PageRank} suggests a way of defining the relatedness of two categories on \textsc{odp} using their localized PageRank scores. The goal is to generalize the idea of the least-common ancestor to random walks to give a different sense of the distance between categories. To do so, create a graph from the directed hierarchy in the \textsc{odp}. Let $\vx$ be the reverse PageRank vector that teleports back to a single category, and let $\vy$ be the reverse PageRank vector that teleports back to another (single) category. Then the relatedness of these categories is the cosine of the angle between $\vx$ and $\vy$. Let $\vx$ be the localized PageRank vector (Note the use of reverse PageRank here so that edges go from child to parent.) They show evidence that this is a useful measure of relationship in ODP.

\paragraph{Application 4 -- Logic programming} A fundamental challenge with scaling logic programming systems like Prolog is that there is an exponential explosion of potential combinations and rules to evaluate and, unless the system is extremely well-designed, these cannot be pruned away. This limits applications to almost trivial problems. Internally, Prolog-type systems resolve, or prove, logical statements using a search procedure over an implicitly defined graph that may be infinite. At each node of the graph, the proof system generates all potential neighbors of the node by applying a rule set given by the logic system. Thus, given one node in the graph, the search procedure eventually visits all nodes. Localized PageRank provides a natural way to restrict the search space to only ``short'' and ``likely'' proofs~\cite{Wang-2013-ProPPR}. Formally, they use PageRank's random teleportation to control the expansion of the search procedure. However, there is an intuitive explanation for the random restarts in such a problem: periodically we all abandon our current line of attack in a proof and start out fresh. Their system with localized PageRank allows them to realize this behavior in a rigorous way.

\subsection{PageRank in recommender systems: ItemRank} \label{sec:recommender}

A recommender system attempts to predict what its users will do based on their past behavior. Netflix and Amazon have some of the most famous recommendation systems that predict movies and products, respectively, their users will enjoy. Localized PageRank helps to score potential predictions in many research studies on recommender systems.

\paragraph{Query reformulation}
A key component of modern web-search systems is predicting future queries. \Citet{Boldi-2008-QueryFlow} run localized PageRank on a query reformulation-graph that
describes how users rewrite queries with $\alpha = 0.85$. Two queries, $q_1$ and $q_2$, are connected in this graph if a user searched for $q_1$ before $q_2$ within a close time-frame and both $q_1$ and $q_2$ have some non-trivial textual relationships. This graph is directed and weighted. The teleportation vector is localized on the current query, or a small set of previously used terms. PageRank has since had great success for many tasks related to query suggestion and often performs among the best methods evaluated~\cite{Song-2012-Query-suggestion}.

\paragraph{Item recommendation: ItemRank}
Both Netflix and Amazon's recommender systems are called \emph{item recommendation} problems. Users rate items -- typically with a 5-star scale -- and we wish to recommend items that a user will rate highly. The ratings matrix is an items-by-users matrix where $R_{ij}$ is the numeric rating given to item $i$ by user $j$. These ratings form a bipartite network between the two groups and we collapse this to a graph over items as follows. Let $G$ be a weighted graph where the weights on an edge $(i,j)$ are the number of users that rated both items $i$ and $j$. (These weights are equivalent to the number of paths of length 2 between each pair of items in terms of the bipartite graph.) Let $\mP$ be the standard weighted random walk construction on $G$. Then the ItemRank scores~\cite{Gori2007-itemrank} are the solutions of: 
\[ (\eye - \alpha \mP) \mS = (1-\alpha) \mR \mD_{\mR}^{-1} \]
where $\mD_{\mR}$ are column sums of the rating matrix. Each column of $\mS$ is a set of recommendations for user $j$, and $S_{ij}$ is a proxy for the interest of user $j$ in item $i$. Note that any construction of the transition matrix $\mP$ based on correlations between items based on user ratings would work in this application as well.

\paragraph{Link prediction}
Given the current state of a network, link prediction tries to predict which edges will come into existence in the future. \Citet{nowell2006-link-prediction} evaluated the localized PageRank score of an unknown edge in terms of its predictive power. These PageRank values were entries in the matrix $(\mI - \alpha \mP)^{-1}$ for edges that currently do not exist in the graph. PageRank with $\alpha$ between 0.5 and 0.99 was not one of their best predictors, but the Katz matrix $(\mI - \alpha \mA)^{-1}$ was one of the best with $\alpha = 0.0005$. Note that Katz's matrix is, implicitly, a pseudo-PageRank problem if $\alpha < \frac{1}{d_{\max}}$ where $d_{\max}$ is the largest degree in the graph. The co-authorship graphs tested seem to have had degrees less than $2000$, making this hidden pseudo-PageRank problem one of the best predictors of future co-authorship. More recent work using PageRank for predicting links on the Facebook social network includes a training phase to estimate weights of the matrix $\mP$ to achieve higher prediction \cite{Backstrom-2011-Supervised-Random-Walks}. Localized PageRank is believed to be part of Twitter's follower suggestion scheme too~\cite{Bahmani-2010-PageRank}.

\subsection{PageRank in social networks: BuddyRank, TwitterRank} \label{sec:social}

PageRank serves three purposes in a social network, where the nodes are people and the edges are some type of social relationship. First, as we discussed in the previous section, it can help solve link prediction problems to find individuals that will become friends soon. Second, it serves a classic role in evaluating the centrality of the people involved to estimate their social status and power. Third, it helps evaluate the potential influence of a node on the opinions of the network. 

\paragraph{Centrality: BuddyRank}
Centrality methods have a long history in social networks  -- see \citet{katz1953-status} and \citet{Vigna-2009-spectral} for a good discussion. 
The following claim is difficult to verify, but we suspect that the first use of PageRank in a large-scale social network was the  BuddyRank measure employed by BuddyZoo in 2003.\footnote{\url{http://web.archive.org/web/20050724231459/http://buddyzoo.com/}}
BuddyZoo collected contact lists from users of the AOL Instant Messenger service and assembled them into one of the first large-scale social networks studied via graph theoretic methods. 
Since then, PageRank has been used to rank individuals in the Twitter network by their importance~\cite{Java2007-Twitter-blog} and to help characterize properties of the Twitter social network by the PageRank values of their users~\cite{Kwak2010-Twitter}. These are standard applications of PageRank with global teleportation and $\alpha \approx 0.85$.

\paragraph{Influence} Finding influential individuals is one of the important questions in social network analysis. This amounts to finding nodes that can spread their influence widely. More formalizations of this question result in NP-hard optimization problems~\cite{Kempe-2003-social-influence} and thus, heuristics and approximation algorithms abound~\cite{Kempe-2003-social-influence,Kempe-2005-influence}. Using Reverse PageRank with global teleportation as a heuristic outperforms out-degree for this task, as shown by \citet{Java-2006-spread} for web-blog influence and \citet{Bar-Yossef2008-reverse-PageRank} for the social network LiveJournal. Reverse PageRank, instead of traditional PageRank, is the correct model to understand the \emph{origins} of influence -- the distinction is much like the treatment of hubs and authorities in other  ranking models on networks~\cite{kleinberg1999-hits,blondel2004-graph-similarity}. These ideas also extend to finding topical authorities in social networks by using the teleportation vector and topic-specific transition probabilities to localize the PageRank vector in TwitterRank~\cite{Weng2010-TwitterRank}.

\subsection{PageRank in the web, redux: HostRank, DirRank, TrustRank, BadRank, VisualRank} \label{sec:web}

At the conclusion of our survey of applications, we return to uses of PageRank on the web itself. Before we begin, let us address the elephant in the room, so to speak. Does Google still use PageRank? Google reportedly uses a basket of ranking metrics to determine the final order that results are returned. These evolve continuously and vary depending on where and when you are searching. It is unclear to what extent PageRank, or more generally, link analysis measures play a role in Google's search ordering, and this is a closely guarded secret unlikely to be known outside of an inner-circle at Google. One the one hand, in perhaps the only large-scale published study on PageRank's effectiveness in a search engine, \citet{najork2007-hits} found that it underperformed in-degree. On the other hand, PageRank is still widely believed to still play some role based on statements from Google. For instance, Matt Cutts, a Google engineer, wrote about how Google uses PageRank to determine crawling behavior~\cite{cutts2006-pagerank-crawling}, and later wrote about how Google moved to a full substochastic matrix in terms of their PageRank vector~\cite{cutts2009-pagerank-scuplting}. The latter case was designed to handle a new class of link on the web called \texttt{rel=nofollow}. This was an optional HTML parameter that would tell a crawler that the following link is not useful for relevance judgements. All the major web companies created this parameter to combat links created in the comment sections of extremely high quality pages such as the Washington Post. These links are created by users of the Washington Post, not the staff themselves, and shouldn't constitute an endorsement on a page. Cutts described how Google's new PageRank equation would count these \texttt{rel=nofollow} links in the \emph{degree} of a node when it was computing a stochastic normalization, but would remove the links when computing relevance.  For instance, if my page had three true links and two \texttt{rel=nofollow} links, then my true links would have probabilities $1/5$ instead of $1/3$, and the sum of my outgoing probability would be $3/5$ instead of 1. Thus, Google's PageRank computation is a pseudo-PageRank problem now. 

Outside of Google's usage, PageRank is also used to evaluate the web at coarser levels of granularity through HostRank and DirRank. Reverse PageRank provides a good measure of a page's similarity to a hub, according to 
both \citet{Fogaras-2003-where} and \citet{Bar-Yossef2008-reverse-PageRank}. PageRank and reverse PageRank also provide information on the ``spaminess'' of particular pages through metrics such as TrustRank and BadRank. PageRank-based information also helped to identify spam directly in a study by \citet{becchetti2008-spam}. Finally, PageRank helps identify canonical images to place on a web-search result (VisualRank).

\paragraph{Coarse PageRank: HostRank, DirRank}
\Citet{arasu2002-pagerank} was an important early paper that defined HostRank, where the web is aggregated at the level of hostnames. In this case, all links to and from a hostname, such as \texttt{www.cs.purdue.edu}, become equivalent. This particular construction models a random surfer that, when visiting a page, makes a censored, or silent, transition within all pages on the same host, and then follows a random link. The HostRank scores are the sums of these modified PageRank scores on the pages within each host~\cite{gleich2007-approximate}. Later work included BlockRank~\cite{kamvar2003-blockrank}, which used HostRank to initialize PageRank, and DirRank~\cite{eiron2004-ranking}, which forms an aggregation at the level of \emph{directories} of websites. 

\paragraph{Trust, Reputation, \famp\ Spam: TrustRank, BadRank}
PageRank typically provides authority scores to estimate the importance of a page on the web. As the commercial value of websites grew, it became highly profitable to create \emph{spam} sites that contain no new information content but attempt to capture Google search results by appearing to contain information. BadRank~\cite{Sobek-2002-PR0} and TrustRank~\cite{gyongyi2004-trustrank} emerged as new, link analysis tools to combat the problem. Essentially, these ideas solve localized, reverse PageRank problems. The results are either used directly, or as a ``safe teleportation'' vector for PageRank, as in TrustRank, or in concert with other techniques, as likely done in BadRank. \Citet{Kolda-2009-BadRank} generalizes these models and includes the idea of adding self-links to fix the dangling nodes, like in sink preferential PageRank, but they add them everywhere, not just at dangling nodes. For spam-link applications, this way of handling dangling nodes is superior -- in a modeling sense -- to the alternatives.

\paragraph{Wikipedia} 
Wikipedia is often used as a subset of the web for studying ranking. It is easy to download the data for the entire website, which makes building the web-graph convenient. (A crawl from a few years ago is in the sparse matrix repository,~\citealt{Davis2011-matrix}, as the matrix \texttt{Gleich/wikipedia-20070206}.) Current graphs of the English language pages have around 100,000,000 links and 10,000,000 articles. The nature of the pages on Wikipedia also makes it easy to evaluate results anecdotally. For instance, we would all raise an eyebrow and demand explanation if ``Gene Golub'' was the page with highest global PageRank in Wikipedia. On the other hand, this result might be expected if we solve a localized PageRank problem around the Wikipedia article  for ``numerical linear algebra.'' \Citet{wissner-gross2006-wikipedia-reading} used Wikipedia as a test set to build reading lists using a combination of localized and global PageRank scores. Later, \citet{Zhirov-2010-wikipedia} computed a 2d ranking on Wikpedia by combining global PageRank and reverse PageRank. Finally, this 2d ranking showed that Frank Sinatra was one of the most important people~\cite{Eom-2014-wikipedia}.

\paragraph{Image search: VisualRank}
PageRank also helps to identify ``canonical'' images to display as a visual summary of a larger set of images returned from an image search engine. In the VisualRank system, \citet{Jing-2008-VisualRank} compute PageRank of an image similarity graph generated from an image search result. The graphs are small -- around 1000 nodes -- which reflects the standard textual query results, and they are also symmetric and weighted. They solve a global PageRank problem with uniform teleportation or high-result biased teleportation. The highest ranked images are canonical images of Mona Lisa amid a diverse collection of views.

\section{PageRank generalizations}
\label{sec:generalizations}
Beyond the applications discussed so far, there is an extremely wide set of PageRank-like models that do not fit into the canonical definition and constructions from Section~\ref{sec:constructions}. These support a wide range of additional applications with mathematics that differs slightly, and some of them are formal mathematical generalizations of the PageRank vectors.   For instance, in prior work, we studied PageRank with a random teleportation parameter~\cite{constantine2010-rapr}. The standard deviation of these vectors resulted in increased accuracy in detecting spam pages on the web. We now survey some of these formal generalizations.

\subsection{Diffusions, damped sums, \famp\ heat kernels}\label{sec:damped}
Recall that the pseudo-PageRank vector is the solution of 
\eqref{eq:pseudo-pr}, 
\[ (\eye - \alpha \mPbar) \vy = \vf. \]
Since all of the eigenvalues of $\mPbar$ are bounded by $1$ in magnitude, the solution $\vy$ has an expansion in terms of the Neumann series: 
\[ \vy = \sum_{k = 0}^\infty \alpha^k \mPbar^k \vf. \]
This expressions gives the pseudo-PageRank vector as a damped sum of powers of $\mPsub$  where each power, $\mPbar^k$, has the geometrically decaying weight $\alpha^k$. These are often called \emph{damped diffusions} because this equation models how the quantities in $\vf$ probabilistically diffuse through the graph where the probability of a path of length $k$ is damped by $\alpha^k$. Many other sequences serve the same purpose as pointed out by a variety of authors.

\paragraph{Generalized damping} Perhaps the most general setting for these ideas is the generalized damped PageRank vector:  
\begin{equation} \label{eq:damped-pr}
 \vz = \sum_{k=0}^\infty \gamma_k \mPbar^k \vf
\end{equation}
where $\gamma_k$ is a non-negative $\ell_1$-sequence (that is, $\sum_k \gamma_k < \infty$ and $\gamma_k \ge 0$). This reduces to PageRank if $\gamma_k = \alpha^k$.
\Citet{huberman1998-surfing-behavior} suggested using such a construction where $\gamma_k$ arises from real-world path following behaviors on the web, which they found to resemble inverse Gaussian functions. Later results from \citet{baeza2006-generalizing} proposed essentially the same formula in~\eqref{eq:damped-pr}. They suggested a variety of interesting functions $\gamma_k$, including some with only a finite number of non-zero terms. These authors drew their motivation from the earlier work of TotalRank~\cite{boldi2005-totalrank}, which suggested $\gamma_k = \frac{1}{k+1} - \frac{1}{k+2}$ in order to evaluate the TotalRank vector: 
\[ \vz = \int_{0}^1 (\eye - \alpha \mPbar)^{-1} (1-\alpha) \vv \, d\alpha. \]
This integrates over all possible values of $\alpha$. (As an aside, this integral is well defined because a unique limiting PageRank value exists at $\alpha = 1$, see Section~\ref{sec:limit}. This sidesteps a technical issue with the singular matrix at $\alpha=1$.) Our work with making the value of $\alpha$ in PageRank a random variable is really a further generalization~\cite{constantine2010-rapr}. Let $\vx(\alpha)$ be a parameterized form for the PageRank vector for a fixed graph and teleportation vector. Let $A$ be a random variable supported on $[0,1]$ with an infinite number of finite moments, that is, $E[A^k] < \infty$ for all $k$. Intuitively, $A$ is the probability that a random user of the web follows a link. Our idea was to use the expected value of PageRank $E[\vx(A)]$ to produce a ranking that reflected the distribution of path-following behaviors in the random surfers. We showed: 
\[ E[\vx(A)] = \sum_{k=0}^\infty (E[A^k] - E[A^{k+1}]) \mP^k \vv. \] This results in a family of sequences of $\gamma_k$ that depend on the random variable $A$. Recent work by \citet{Kollias-2013-multidamping} shows how to evaluate these generalized damped vectors as a polynomial combination of PageRank vectors in the sense of \eqref{eq:pr}.

\paragraph{Heat kernels \famp\ matrix exponentials} Another specific case of generalized damping arises from the matrix exponential, or heat kernel: 
\[ \vz = e^{\beta \mPsub} \vf = \sum_{k=0}^{\infty} \frac{\beta^k}{k!} \mPsub^k \vf. \]
Such functions arose in a wide variety of domains that would be tangential to review here~\cite{Estrada-2000-index,miller2001-expohits,Kondor-2002-diffusion,farahat2006-exphits,chung2007-pagerank-heat,Kunegis-2009-learning-spectral,Estrada-2010-matrix-functions}. In terms of a specific relationship with PageRank, \citet{Yang-2007-DiffusionRank} noted that the pseudo-PageRank vector itself was a single-term approximation to these heat kernel diffusions. Consider 
\[ \vz = e^{\beta \mPsub} \vf \quad \Leftrightarrow \quad  e^{-\beta \mPsub} \vz = \vf \quad  \Leftrightarrow  \quad (\mI - \beta \mPsub + \ldots) \vz = \vf. \]
If we truncate the heat kernel expansion after just the first two terms ($\mI - \beta \mP$), then we arise at the pseudo-PageRank vector. (A similar result holds for the formal PageRank vector too.)

\subsection{PageRank limits \famp\ eigenvector centrality} \label{sec:limit} 
In the definition of PageRank used in this paper, we assume that $\alpha < 1$. PageRank, however, has a unique well-defined limit as $\alpha \to 1$~\cite{serra2005-jcf-pagerank,boldi2005-damping,boldi2009-functional}. This is easy to prove using the Jordan canonical form for the case of PageRank~\eqref{eq:pr}, but extensions to pseudo-PageRank are slightly more nuanced. As in the previous section, let $\vx(\alpha)$ be the PageRank vector as a function of $\alpha$ for a fixed stochastic $\mP$: $(\mI - \alpha \mP) \vx(\alpha) = (1-\alpha) \vv$. Let $\mX \mJ \mX^{-1}$ be the Jordan canonical form of $\mP$. Because $\mP$ is stochastic, it's eigenvalues on the unit circle are all semi-simple~\cite[page 696]{meyer2000-book}. Thus: 
\[ \mJ = \sbmat{\mI \\ & \mD_1 \\ & & \mJ_2}, \]
where $\mD_1$ is a diagonal matrix of the eigenvalues on the unit circle and $\mJ_2$ is a Jordan block for all eigenvalues with $|\lambda| < 1$. We now substitute this into the PageRank equation: 
\[ (\eye - \alpha \mP) \vx(\alpha) = (1-\alpha) \vv \Leftrightarrow  (\eye - \alpha \mJ)^{-1} \underbrace{\vec{\hat{x}}(\alpha)}_{\mX^{-1} \vx(\alpha)} = (1-\alpha) \underbrace{\vec{\hat{v}}}_{\mX^{-1} \vv}. \]
Using the structure of $\mJ$ decouples these equations: 
\[ \left( \sbmat{\mI \\ & \mI \\ & & \mI} - \alpha \sbmat{\mI \\ & \mD_1 \\ & & \mJ_2} \right) \sbmat{ \vec{\hat{x}}(\alpha)_0 \\ \vec{\hat{x}}(\alpha)_1 \\ \vec{\hat{x}}(\alpha)_2 } = (1-\alpha) \sbmat{ \vec{\hat{v}}_0 \\ \vec{\hat{v}}_1 \\ \vec{\hat{v}}_2 }. \]
As $\alpha \to 1$, both $\vec{\hat{x}}(\alpha)_1$ and $\vec{\hat{x}}(\alpha)_2$ go to $0$ because these linear systems remain non-singular. Also, note that $\vec{\hat{x}}(\alpha)_0 = \vec{\hat{v}}_1 $ for all $\alpha \not=1$, so this point is a removable singularity. Thus, $\vec{\hat{x}}$ can be uniquely defined at $\alpha = 1$, and hence, so can $\vx$. \Citet{vigna2005-trurank} uses the structure of this limit to argue that taking $\alpha \to 1$ in practical applications is not useful unless the underlying graph is strongly connected, and they propose a new PageRank construction to ensure this property. Subsequent work by \citet{Vigna-2009-spectral} does a nice job of showing how limiting cases of PageRank vectors converge to traditional eigenvector centrality measures from bibliometrics~\cite{Pinski-1976-influence} and social network analysis~\cite{katz1953-status}.  

The pseudo-PageRank problem does not have nice limiting properties in our formulation. Let $\vy(\alpha)$ be a parametric form for the solution of the pseudo-PageRank system $(\eye - \alpha \mPsub) \vy = \vf$. As $\alpha \to 1$, then $\vy \to \infty$, unless the non-zero support of $\vf$ lies outside of a recurrent class, in which case $\vy \to 0$. \Citet{boldi2005-damping} defines the PseudoRank system as: 
\[ (\eye - \alpha \mPbar) \vy = (1-\alpha) \vf \]
instead. This system always has a non-infinite limit as $\alpha \to 1$. It could, however, have zero as a limit if $\mPbar$ has all eigenvalues less than 1. 

\subsection{Over-teleportation, negative teleportation, \famp\ the Fiedler vector}
\label{sec:negative}
The next generalization of PageRank is to values of $\alpha > 1$. These arose in our prior work to understand the convergence of quadrature formulas for approximating the expected value of PageRank with random teleportation parameters~\cite{constantine2010-rapr}. \Citet{Mahoney-2012-local} subsequently showed an amazing relationship among (i) the Fiedler vector of a graph~\cite{Fiedler1973-algebraic-connectivity,Anderson-1985-Laplacian,Pothen-1990-partitioning}, (ii) a particular generalization of the PageRank vector, which we call MOV, and (iii) values of $\alpha > 1$.

\paragraph{The Fiedler vector} {In contrast to the remainder of this paper, the constructions and statements in this section are specific to connected, undirected graphs with symmetric adjacency matrices.}  The conductance of a set of vertices in a graph is defined as the number of edges leaving that set, divided by the sum of the degrees of the vertices within the set. Conductance and its relatives are often used as numeric quality scores for graph partitioning in parallel processing~\cite{Pothen-1990-partitioning} and for community detection in graphs~\cite{Schaeffer-2007-clustering}. It is NP-hard to find the set of smallest conductance, but Fiedler's vector reveals information about the presence of small conductance sets in a graph through the Cheeger inequality~\cite{Chung-1992-book}. Let $G$ be a connected, undirected graph with symmetric adjacency matrix $\mA$ and diagonal degree matrix $\mD$. The Fiedler vector is the generalized eigenvector of $(\mD - \mA) \vq = \lambda_* \mD \vq$, with the smallest positive eigenvalue $\lambda_* > 0$. All of the generalized eigenvalues are non-negative, the smallest is $0$, and the largest is bounded above by 1.  Cheeger's inequality bounds the relationship between $\lambda_*$ and the set of smallest conductance in the graph. 
    
\paragraph{MOV} The MOV vector is defined as the pseudo-inverse solution $\vr$ in the consistent linear system of equations: 
\begin{equation}
[(\mD - \mA) - \gamma \mD] \vr = \rho(\gamma) \mD \vs,
\end{equation}
where $\gamma < \lambda_*$, $\vs$ is a ``seed'' vector such that $\vs^T \mD \ve = 0$, and $\rho(\gamma)$ is a scaling constant such that $\vr$ has a fixed norm. When $\gamma = 0$, this system is singular but consistent, and thus, we take the pseudo-inverse solution. Note that this is equivalent to the pseudo-PageRank problem: 
\[ (\mI - \alpha \mP) \vz = \alpha \rho(\gamma) \vfhat \]
where $\alpha = \frac{1}{1-\gamma}$, $\vz = \mD \vr$, and $\vfhat = \mD \vs$.  The properties of $\vs$ in MOV imply that $\vfhat^T \ve = 0$, and thus, $\vfhat$ must have negative elements, which generalizes the standard pseudo-PageRank. 

In a small surprise, allowing $\vf$ to take on negative values results in no additional modeling power in the case of symmetric $\mA$. To establish this result, we first observe that: 
\[ (\eye - \alpha \mA \mD^{-1}) \tfrac{\sigma}{1-\alpha} \vd = \sigma \vd. \] This preliminary fact shows that the pseudo-PageRank vector of an undirected graph with teleportation according to the degree vector $\vd$ simply results in a rescaling. We can use this property to shift any $\vf$ with negative values in a controlled manner: 
\[ (\eye - \alpha \mP) \underbrace{(\vz + \tfrac{\sigma}{1-\alpha} \vd)}_{\vy} = \underbrace{\alpha \vfhat + \sigma \vd}_{\vf}, \]
where $\sigma$ is chosen such that $\vf \ge 0$ element-wise. Solving these shifted pseudo-PageRank systems, then, effectively computes the solution $\vz$ with a well-understood bias term $\theta \vd$. This is easy to remove afterwards: $\vz = \vy - \frac{\sigma}{1-\alpha} \vd$, at which point we can normalize $\vz$ to account for $\rho(\gamma)$ if desired.

\paragraph{Values of $\alpha > 1$} While this generalization with negative entries in $\vf$ gives no additional mathematical power, it does permit a seamless limit from PageRank vectors to the Fiedler vector. Let $\alpha_* = \frac{1}{1-\lambda_*} > 1$. The formal result is that the limit $\lim_{\alpha \to \alpha_*} \frac{1}{\rho(\alpha)} \vz(\alpha) = \vq$, the Fiedler vector. Note that for the construction of $\mP = \mA \mD^{-1}$ on an undirected, connected graph, we have that $\mP^k \to \frac{1}{\ve^T \vd} \vd \ve^T$ as $k \to \infty$. Thus, when $\alpha = 1$, the MOV solution $\vz$ is equivalent to the solution of $(\eye - (\mP - \frac{1}{\ve^T \vd} \vd \ve^T) ) \vz = \vf$ because the right hand side of $\vf$ is orthogonal to the left eigenvector $\ve^T$.  As all of the eigenvalues of $(\mP - \frac{1}{\ve^T \vd} \vd \ve^T)$ are distinct from $1$, this is a non-singular system. And this fact allows the limit construction to pass through $\alpha = 1$ seamlessly. If we additionally assume that $\vf^T \vq \not= 0$, then 
\[ \lim_{\alpha \to \alpha_*} \frac{1}{\rho(\alpha)} \vz(\alpha) = \vq, \]
and the limiting value of PageRank with over-teleportation is the Fiedler vector. The analysis in \citet{Mahoney-2012-local}, then, interpolates many of the arguments in  \citet{Vigna-2009-spectral} beyond $\alpha = 1$ to yield important relationships between spectral graph theory and PageRank vectors.


\subsection{Complex-valued teleportation parameters and a time-dependent generalization} \label{sec:complex}
Again, let $\vx(\alpha)$ be the PageRank vector (in the sense of \eqref{eq:pr}) as a function of $\alpha$ for a fixed graph and teleportation vector. Mathematically, the PageRank vector is a rational function of $\alpha$. This simple insight produces a host of possibilities, one of which is evaluating the derivative of the PageRank vector~\cite{boldi2005-damping,greif2006-arnoldi,gleich2007-pagerank-deriv}. Another is that PageRank with complex-valued $\alpha$ is a reasonable mathematical generalization~\cite{horn2008-parametric-google}. Let $\alpha \in \CC$ with $|\alpha| < 1$, then $\vx(\alpha)$ has some interesting properties and usages. In \citet{constantine2010-rapr}, we needed to bound $\normof[1]{\vx(\alpha)}$ when $\alpha$ was complex. If $\alpha$ is real and $0 < \alpha < 1$, then $\normof[1]{\vx(\alpha)} = 1$ independent of the choice of $\alpha$. However, if $\alpha$ is complex we have: $\normof[1]{\vx} \le \frac{|1-\alpha|}{1-|\alpha|}.$ Later, in \citet{Gleich-2014-dynamic-pagerank}, we found that complex values of $\alpha$ arise in computing closed form solutions to PageRank dynamical systems where the teleportation vector is a function of time, but the graph remains fixed. Specifically, the PageRank vector with complex teleportation arises in the steady-state time-dependent solution of
\[ \vx'(t) = (1-\alpha) \vv(t) - (\eye - \alpha \mP) \vx(t), \] when $\vv(t)$ oscillates between a fixed set of vectors. Thus, PageRank with complex teleportation is both an interesting mathematical problem and has practical applications in a time-dependent generalization of PageRank.

\subsection{Censored node constructions} \label{sec:dummy-node} \label{sec:censored-node}
The final generalized PageRank construction we wish to discuss is, in fact, a PageRank system hiding inside a Markov chain construction with a different type of teleportation. In order to motivate the particular form of this construction, we first review an alternative derivation of the PageRank vector.

A censored node in a Markov chain is one that exhibits a virtual influence on the chain in the sense that walks proceed through it as if it were not present. Let us illustrate this idea by crafting teleportation behavior into a Markov chain in a different way and computing the PageRank vector itself by censoring that Markov chain. Suppose that we want to find the stationary distribution of a walk where, if a surfer wants to teleport, they first transition to a teleport state, and then move from the teleport state according to the teleportation distribution. The transition matrix of the Markov chain is: 
\[ \mP' = \bmat{ \alpha \mP & \vv \\ (1-\alpha) \ve^T & 0 }. \]
And the stationary distribution of this Markov chain is: \[\bmat{ \alpha \mP & \vv \\ (1-\alpha) \ve^T & 0 } \bmat{ \vx' \\ \gamma } = \bmat{ \vx' \\ \gamma}, \qquad \ve^T \vx' + \gamma = 1. \]
Censoring the final teleportation state amounts to modeling its influence on the stationary distribution, but leaving it with no final contribution. Put more formally, the stationary distribution of the censored chain is just $\vx'$ renormalized to be a probability distribution: $\vx = \vx' / \ve^T \vx'$. In other words, censoring that state models pretending that it wasn't there when determining the stationary distribution, but the transitions through it still took place; this is equivalent to the standard teleporting behavior. The vector $\vx$ is also the PageRank vector of $\alpha, \mP, \vv$, which follows from 
\[ \vx = \tfrac{1-\alpha}{\gamma} \vx' = \tfrac{1-\alpha}{\gamma} \left[ \alpha \mP \vx' + \gamma \vv \right] = \alpha \mP \vx + (1-\alpha) \vv. \]
\Citet{Tomlin-2003-traffic-rank}, \citet{eiron2004-ranking} and \citet[written in 2003]{lee2007-fast} were some of the first to observe this property in the context of PageRank; although censoring Markov chains goes back much further.

There is a more general class of PageRank-style methods that craft transitions akin to non-uniform teleportation through a censored node construction. Consider, for example, adding a teleportation node $c$ that connects to all nodes of a network as in Figure~\ref{fig:dummy}. This construction gives rise to an implicit PageRank problem with $\alpha = \frac{d_{max}}{d_{max}+1}$ as we now show. Let 
\[ \mA' = \bmat{\mA & \ve \\
                \vv^T & 0 } \]
be the adjacency matrix for the modified graph, where $\vv$ is the teleportation destination vector. A uniform random walk on this adjacency structure has a single recurrent class, and thus, a unique stationary distribution~\cite[Theorem 3.23]{berman1987-nonnegative}. The stationary distribution satisfies:
\[ \mP' \vx = \vx \quad \Leftrightarrow 
\bmat{ \mA^T (\mD + \mI)^{-1} & \vv / \ve^T \vv \\
                 \ve  (\mD + \mI)^{-1}  & 0 }
\bmat{ \vx' \\ \gamma } = \bmat{ \vx' \\ \gamma }. \]
Let $\mPsub' = \mA^T (\mD + \mI)^{-1}$. The censored distribution $\vx = \vx' / \ve^T \vx' $ is a normalized solution of the linear system: 
\begin{equation} \label{eq:dummy-node-pr}
(\mI - \mPsub') \vx = \vv. 
\end{equation}
Note that $\vc^T = \ve^T - \ve^T \mPsub' > 0$, and so all columns are substochastic. This means that all of the nodes ``leak probability'' in a semi-formal sense. Scaling $\mPsub'$ by $\frac{1}{1-c_{\max}} > 1$ adjusts the probabilities such that there is at least one column that is stochastic. Consequently, we can write $\mPsub' = \alpha \mPsub$ where $\alpha = (1-c_{\max})$ and $\mPsub = \frac{1}{1-c_{\max}} \mPsub'$. By substituting this form into~\eqref{eq:dummy-node-pr}, we have that $\vx$ is the normalized solution of a pseudo-PageRank problem where $\alpha = \frac{1}{1-c_{\max}}$. Assuming that $\mA$ is an unweighted graph, then $\alpha = \frac{d_{\max}}{d_{\max}+1}$.

\begin{figure}
\footnotesize
\begin{minipage}[m]{0.3\linewidth}
\centering
\includegraphics{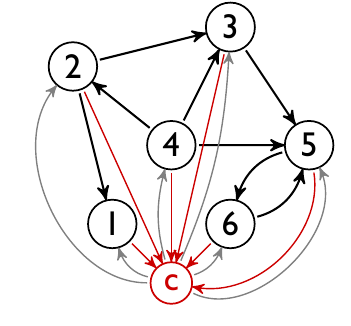}
\end{minipage}%
\begin{minipage}[m]{0.7\linewidth}
\centering
$\mPsub = \bmat{0 & 1/3 & 0   & 0   & 0 & 0 \\
                0 & 0   & 0   & 1/4 & 0 & 0 \\
                0 & 1/3 & 0   & 1/4 & 0 & 0 \\
                0 & 0   & 0   & 0   & 0 & 0 \\
                0 & 0   & 1/2 & 1/4 & 0 & 1/2 \\
                0 & 0   & 0   & 0   & 1/2 & 0}
\quad 
\vc = \bmat{ 1 \\ 1/3 \\ 1/2 \\ 1/4 \\ 1/2 \\ 1/2 }$
\end{minipage} \medskip 

\begin{minipage}[m]{0.3\linewidth}
\centering
(a) A directed graph with a censored node $c$
\end{minipage}%
\begin{minipage}[m]{0.7\linewidth}
\centering
(b) The substochastic matrix and correction vector for the Markov chain construction after node $c$ is censored.
\end{minipage}
\caption{In this teleportation construction we add a node $c$ to the original graph as in subfigure (a). The probability of transitioning to $c$, or teleporting after we censor node $c$, then depends on the degree of each node. A random surfer teleports from node $2$ with probability $1/3$ and from node $4$ with probability $1/4$. This construction yields a substochastic matrix $\mPsub$ where all the elements of the correction vector $c$ are positive. This means it's equivalent to a PageRank construction with $\alpha = 1-\min \vc$, or $\alpha = 3/4$ for this problem. }
\label{fig:dummy}
\end{figure}

This idea frequently reappears; for instance, \citet{Bini-2010-combined}, \citet{Lu-2011-LeaderRank} and \citet{Schlote-2012-road-rank} all use it in different contexts.

In a different context, this same type of analysis shows that the Colley matrix for ranking sports teams is a diagonally perturbed, generalized pseudo-PageRank system~\cite{Colley-2002-bias,Langville-2012-book}. Let the symmetric, weighted graph $G$ represent the network of times team $i$ played team $j$. And let $\vf$ be a vector of the accumulated scores differences over all of those games. It could have negative entries, rendering it outside of our traditional framework, however, as we saw in Section~\ref{sec:negative}, this is a technical detail that is avoidable. The vector of Colley scores $\vr$ is the solution of: 
\[ (\mD + 2\mI - \mA) \vr = \vf. \]
Let $\vy = (\mD+2\mI)^{-1} \vr$. Then, 
\[ (\mI - \alpha \mPsub) \vy = \vf \]
where $\alpha = \frac{d_{\max}}{d_{\max}+2}$.
This analysis establishes a formal relationship between Markov style ranking metrics~\cite{Langville-2012-book} and the least-squares style ranking metrics employed by Colley. It also enables us to use fast PageRank solvers for these Colley systems.

\section{Discussion \famp\ a positive outlook on PageRank's wide usage} \label{sec:conclusion}

PageRank has gone from being used to evaluate the importance of web pages to a much broader set of applications. The method is easy to understand, is robust and reliable, and converges quickly. Most applications solve PageRank problems of only a modest size, with fewer than 100,000,000 vertices; this regime permits a much wider variety of algorithmic possibilities than those that must only work on the web. 

We have avoided the discussion of PageRank algorithms entirely in this manuscript because, by and large, simple iterations suffice for fast convergence in this regime. Values of $\alpha$ tend to be less than $0.99$, which requires fewer than $2000$ iterations to converge to machine precision. Nevertheless, there is ample opportunity to accelerate PageRank computations in this regime as there are ideas that involve computing \emph{multiple} PageRank vectors for a single task. One example is PerturbationRank~\cite{Du-2008-PerturbationRank}, which uses the perturbation induced in a PageRank vector by removing a node to compute a new ranking of all nodes. Thus, innovations in PageRank algorithms are still relevant, but must be made within the context of these small-scale uses. 

There are also a great number of PageRank-like ideas outside of our specific canon. For instance, none of the following models fit our PageRank framework: 
\begin{description}
\item [BrowseRank] \Citet{liu2008-browserank} define a continuous time Markov chain to model a random surfer that \emph{remains} on a specified node for some period of time before transitioning away. This model handles sites like Facebook, where users spend significant time interacting within a single page. 
\item[Voting] \Citet{boldi2009-voting} and \citet{Boldi-2011-viscous} define a voting model on a social network inspired by computing Katz or PageRank on a random network where each node picks \emph{a single outlink}. 
\item[SimRank] This problem is another way to use PageRank-like ideas to evaluate similarity between the nodes of a graph (like the IsoRank problem)~\cite{jeh2002-simrank}. SimRank, however, involves solving a linear system on a \emph{row sub-stochastic matrix}.
\item[Food webs] The food web is a network where species are linked by the feeding relationships.  \Citet{allesina2009-pagerank-foodwebs} point out a few modifications to PageRank to make it more appropriate. First, they use teleportation to model a constant loss of nutrients from higher-level species and reinject these nutrients through primary producers (such as bacteria). Second, they note that the flow of importance ought to be reversed so that species $i$ points to species $j$ if $i$ is important for $j$'s survival. The result is an eigenvector computation on a fully stochastic matrix.
\item[Opinion dynamics] Models of opinion formation on social network posit strikingly similar dynamics to a PageRank iteration~\cite{Friedkin-1990-opinion,Friedkin-1999-opinion}. The essential difference is that a node's opinion is the average of its in-links, instead of propagating its value to its out-links. Like SimRank, this results in a row sub-stochastic iteration.
\end{description}
The details and implications of these models are fascinating, and this manuscript would double in size if we were to treat them.

In most successful studies, PageRank is used as a type of baseline measure. Its widespread success above extremely simple baselines suggests that its modified random walk is a generally useful alternative worth investigating. In this sense, it resembles a form of regularization. And this is how we feel that PageRank should be used. Note that studies must use care when determining the type of PageRank construction -- weighted, reverse, Dirichlet, etc. -- as this can make a large difference in the quality of the results. Consider, for instance, the use of weighted PageRank in \citet{Jiang-2009-ranking-spaces}. In their application, they wanted to model where people move, and it makes good sense that businesses would locate in places with many connections and therefore, that people would preferentially move to these same locations. Given the generality of the idea and its intuitive appeal, we anticipate continued widespread use of the PageRank idea over the next 20 years in new and exciting applications as network data continues to proliferate. 

\bibliographystyle{dgleichurlshortlinktitles}
\bibliography{../../../bibliography/all-bibliography}

\section*{Acknowledgments}
We acknowledge the following individuals for their discussions about this manuscript: Sebastiano Vigna, Amy Langville, Michael Saunders, Chen Greif, Des Higham, and Stratis Gallopoulos, as well as Kyle Kloster for carefully reading several early drafts. This work was supported in part by NSF CAREER award CCF-1149756.

\end{document}